\documentclass[11pt]{amsart}
\usepackage[margin=1in]{geometry}
\usepackage{amsfonts}
\usepackage{amssymb}
\usepackage[dvips]{graphics}
\usepackage{epsfig}
\usepackage{euscript}
\usepackage{color}

 \newtheorem{thm}{Theorem}[section]
 \newtheorem{cor}[thm]{Corollary}
 \newtheorem{lem}[thm]{Lemma}
 \newtheorem{rem}[thm]{Remark}
 
 \theoremstyle{definition}
 
 \theoremstyle{remark}

 \numberwithin{equation}{section}

 \def\idty{{\mathchoice {\mathrm{1\mskip-4mu l}} {\mathrm{1\mskip-4mu l}} %
{\mathrm{1\mskip-4.5mu l}} {\mathrm{1\mskip-5mu l}}}}

\newcommand{\R}{{\mathbb R}}

\newcommand{\supp}{\operatorname{supp}}

\newcommand{\om}{\omega}
\newcommand{\la}{\lambda}

\newcommand{\abs}[1]{\left\vert#1\right\vert}

\newcommand{\be}{\begin{equation}}
\newcommand{\ee}{\end{equation}}

\begin{document}

\renewcommand{\thefootnote}{\fnsymbol{footnote}}
\title{Dispersive Estimates for Harmonic Oscillator Systems}

\author[V. Borovyk]{Vita Borovyk}
\address{Department of Mathematical Sciences\\
University of Cincinnati\\
Cincinnati, OH 45221, USA}
\email{vita.borovyk@uc.edu}

\author[R. Sims]{Robert Sims}
\thanks{R.\ S.\ was supported in part by NSF grants DMS-0757424 and DMS-1101345}
\address{Department of Mathematics\\
University of Arizona\\
Tucson, AZ 85721, USA}
\email{rsims@math.arizona.edu}

\date{Version: \today }
\maketitle
\bigskip
\begin{abstract}
We consider a large class of harmonic systems, each defined as a quasi-free dynamics on
the Weyl algebra over $\ell^2( \mathbb{Z}^d)$. In contrast to recently obtained,
short-time locality estimates, known as Lieb-Robinson bounds, we prove a number
of long-time dispersive estimates for these models.
\end{abstract}

\maketitle

\footnotetext[1]{Copyright \copyright\ 2011 by the authors. This
paper may be reproduced, in its entirety, for non-commercial
purposes.}

%
%
%
%

\section{Introduction}\label{sec:intro}

In statistical mechanics, a central object of study is the time evolution, or dynamics, 
corresponding to Hamiltonian systems. Those models
with infinitely many degrees of freedom, often thought of as particles, are of
particular interest since they provide a setting in which one may
investigate the development of macroscopic non-equilibrium phenomena
through microscopic details of the system. One of the most well-studied models
is that of a coupled system of harmonic oscillators. For systems of this
type, particles are situated at the sites of a lattice, and they are allowed
to interact through linear forces. The goal of the present work
is to study the long-time behavior of an infinite volume harmonic dynamics
and provide a particular type of dispersive estimate. 

Our interest in these dispersive estimates stems from a wealth of recent
work on locality bounds for general non-relativistic systems; in addition to the 
references included below, see \cite{NS3, hastings2010, sims2010} for recent review articles.
A number of important physical systems are governed by a non-relativistic 
Hamiltonian dynamics, e.g. models of magnetism, lattice oscillators, and 
a variety of complex networks.  It is inherent in models of this type that there
is no strict equivalent of a finite speed of light. Despite this fact, in 1972
Lieb and Robinson, see \cite{lieb1972}, proved that one can associate a finite group 
velocity to the dynamics of many such systems. 

Let us briefly describe their result. Consider a quantum spin
system, see \cite{bratteli1997} for a more thorough development, 
defined over $\mathbb{Z}^d$. In other words, to each $x \in \mathbb{Z}^d$
associate a finite dimensional Hilbert space $\mathcal{H}_x$. For any
finite set $\Lambda \subset \mathbb{Z}^d$, define a composite Hilbert space and 
algebra of local observables by setting
\begin{equation}
\mathcal{H}_{\Lambda} = \bigotimes_{x \in \Lambda} \mathcal{H}_x \quad \mbox{and} \quad \mathcal{A}_{\Lambda} = \bigotimes_{x \in \Lambda} \mathcal{B}( \mathcal{H}_x) \, ,
\end{equation}
where $\mathcal{B}( \mathcal{H})$ is the set of bounded linear operators over $\mathcal{H}$.
Due to the tensor product structure, it is clear that if $\Lambda_0 \subset \Lambda$, then
any $A \in \mathcal{A}_{\Lambda_0}$ can be identified with $A' = A \otimes \idty_{\Lambda \setminus \Lambda_0} \in \mathcal{A}_{\Lambda}$
and so we may regard $\mathcal{A}_{\Lambda_0} \subset \mathcal{A}_{\Lambda}$.
A Hamiltonian $H_{\Lambda}$ on $\mathcal{H}_{\Lambda}$ is a densely defined self-adjoint operator. Self-adjointness guarantees the
existence of a dynamics, or time evolution, $\tau_t^{\Lambda}$, defined by setting
\begin{equation}
\tau_t^{\Lambda}(A) = e^{itH_{\Lambda}}A e^{-itH_{\Lambda}} \quad \mbox{for all } t \in \mathbb{R} \quad \mbox{and} \quad A \in \mathcal{A}_{\Lambda}.
\end{equation} 
Now, fix $X$ and $Y$ finite, disjoint subsets of $\mathbb{Z}^d$ and take $\Lambda \subset \mathbb{Z}^d$ finite with
$X \cup Y \subset \Lambda$. Lieb and Robinson proved that for a large class of Hamiltonians, defined in terms of essentially short-range
interactions, and any decay rate $\mu>0$, there exist numbers $C$ and $v$ such that the bound  
\begin{equation} \label{eq:lrb}
\left\| \left[ \tau_t^{\Lambda}(A), B \right] \right\| \leq C \| A\| \| B \| \min[|X|,|Y|] e^{- \mu \left( d(X,Y) - v|t| \right)} \, 
\end{equation}
holds for all $A \in \mathcal{A}_X$, $B \in \mathcal{A}_Y$, and all $t \in \mathbb{R}$. Here $|X|$ denotes the
cardinality of the set $X$. Optimizing $v$ over the decay rate $\mu$
produces a number, often called the system's Lieb-Robinson velocity, which can, in general, be 
bounded in terms of an appropriate norm on the interaction.  

Since $X$ and $Y$ are disjoint, the tensor product structure of the observable algebra implies that $[ \tau^{\Lambda}_0(A), B] = [A,B] =0$. 
The estimate in (\ref{eq:lrb}) above then demonstrates that for times $t$ with $|t| \leq d(X,Y) /v$ this commutator remains exponentially small.
In this case, disturbances do not propagate through the system arbitrarily fast. Moreover, bounds of the form (\ref{eq:lrb}) can
be used to show that, despite the fact that the system is non-relativistic, dynamically evolved local observables can be well approximated by strictly local 
observables; at least for small times. These ideas have motivated a number of important improvements on the original Lieb-Robinson bounds 
\cite{nach12006, hast2006, nach22006, nachtergaele2009, PS09, Pou10, schuch2010, nach2011}, and these extensions
have proven useful in a variety of applications 
\cite{hastings2004a, NS2, hastings2007, nachtergaele2010, HS, bravyi:2010a, bravyi:2010b, bachmann2011}. 

It is clear from (\ref{eq:lrb}) that Lieb-Robinson bounds only provide relevant information for small times.
In fact, for large times the Lieb-Robinson bound grows exponentially, however, a naive commutator
estimate shows that
\begin{equation} \label{eq:aprioriest}
\left\| \left[ \tau_t(A), B \right] \right\| \leq 2 \| A \| \| B \|
\end{equation}
for any observables $A$ and $B$ and all $t \in \mathbb{R}$. Here we are only using the fact that
the time-evolution is an automorphism. There are a number of important physical systems
where one expects the norms of commutators, as above, to be small for large times; not
because of local effects, but rather because the dynamics is dispersive. The goal of this article 
is to investigate a class of models where the long time behavior of these Lieb-Robinson type 
commutators can be analyzed.

There is a vast literature concerning the dynamics of models corresponding to harmonic oscillators, however,
questions concerning their locality properties are more modern.
In fact,  analogues of Lieb-Robinson type bounds were first proven for certain classical oscillator 
systems in \cite{marchioro1978}, see also \cite{butta2007} and \cite{raz2009} for newer developments.
Locality estimates for quantum models are much more recent, see \cite{eisert2008} for general harmonic
systems and  \cite{nachtergaele2009, amour2009, nachtergaele2010} for models which
allow for anharmonicities. Due to technical difficulties associated with unbounded operators, 
the results one obtains in this context are valid for only a restricted class of observables. 
For example, in \cite{nachtergaele2009} (see also \cite{nachtergaele2010}), the bound
\begin{equation} \label{eq:lrb2}
\left\| \left[ \tau_t(W(f)), W(g) \right] \right\| \leq C \| f \|_1 \| g \|_1 \min[ |X|, |Y|]  e^{- \mu(d(X,Y) - v|t|)}
\end{equation}
is proven to be valid for any Weyl operators $W(f)$ and $W(g)$ corresponding to functions
$f$ and $g$ in $\ell^1(\mathbb{Z}^d)$. Here it is assumed that $f$ and $g$ have 
disjoint supports $X$ and $Y$ respectively, at least one of which being finite; 
see Section~\ref{sec:M+R} for a more detailed discussion of Weyl operators.
In this paper, we will demonstrate that commutators as in (\ref{eq:lrb2}) also 
decay for large time, i.e. we will prove dispersive estimates for certain harmonic lattice models.
 
The paper is organized as follows. In Section~\ref{sec:M+R}, we introduce a class of harmonic models as a 
quasi-free dynamics on the Weyl algebra $\mathcal{W}( \ell^2( \mathbb{Z}^d))$. This setting allows us
to analyze these models directly in the infinite volume context. After discussing relevant
notation (Sections 2.1 and 2.2), we state our main theorems in Section~\ref{sec:mainresults}. The proofs of
these results follow in Sections \ref{sec:uniest} and \ref{sec:main}.

%
%
%
%

\section{Models and Results}\label{sec:M+R}

In this section, we make precise the results of this paper.
We begin, in Section~\ref{subsec:waqd}, with a brief discussion of 
Weyl algebras and the notion of a quasi-free dynamics on a Weyl algebra.
Next, in Section~\ref{subsec:models}, we introduce a large class of
harmonic models on the Weyl algebra over
$\ell^2( \mathbb{Z}^d)$. In Section~\ref{sec:mainresults}, we state our main results, namely
Theorems~\ref{thm:uniest}--\ref{thm:main}. We conclude this section with
some comments and remarks.

%
%
%

\subsection{Weyl algebras and a quasi-free dynamics} \label{subsec:waqd}

We now briefly introduce some relevant notation, and refer the interested reader to \cite{manuceau1973} and 
also \cite{bratteli1997} for a broader discussion. 

A Weyl algebra, or CCR algebra (for canonical commutation relations), can be defined
over any real linear space $\mathcal{D}$ equipped with a symplectic, non-degenerate
bilinear form $\sigma$. This means that, in addition to bilinearity, 
$\sigma : \mathcal{D} \times \mathcal{D} \to \mathbb{R}$ satisfies:
\begin{equation} \label{eq:symp}
\sigma(f,g) = - \sigma(g,f) \quad \mbox{for all } f, g \in \mathcal{D} ,
\end{equation}
and if $\sigma(f,g) = 0$ for all $f \in \mathcal{D}$, then $g = 0$. The Weyl algebra
over $\mathcal{D}$, which we will denote by $\mathcal{W}( \mathcal{D})$, is then
defined to be a $C^*$-algebra generated by Weyl operators, i.e., non-zero elements $W(f)$, 
associated to each $f \in \mathcal{D}$, which satisfy
\begin{equation} \label{eq:invo}
W(f)^* = W(-f) \quad \mbox{for each } f \in \mathcal{D} \, ,
\end{equation}
and
\begin{equation} \label{eq:weylrel}
W(f) W(g) = e^{-i \sigma(f,g)/2} W(f+g) \quad \mbox{for all } f, g \in \mathcal{D} \, .
\end{equation}
It is well known that a $C^*$-algebra generated
by these Weyl operators with the property that $W(0) = \idty$, $W(f)$ is unitary
for all $f \in \mathcal{D}$, and $\| W(f) - \idty \| = 2$ for all $ f \in \mathcal{D} \setminus \{0 \}$ is unique up to $*$-isomorphism, 
see e.g. \cite{bratteli1997}, Theorem 5.2.8. 

A further consequence of Theorem 5.2.8 of \cite{bratteli1997} is that
certain mappings on $\mathcal{D}$ generate evolutions on $\mathcal{W}( \mathcal{D})$.
Specifically, any group of real linear, symplectic transformations $\{ T_t \}_{t \in \mathbb{R}}$, i.e. 
for each $t \in \mathbb{R}$, $T_t:\mathcal{D}\to\mathcal{D}$ and 
\begin{equation} \label{eq:sympT}
\sigma(T_t f, T_t g)= \sigma(f,g)\, ,
\end{equation}
generates a unique dynamics $\tau_t$ on $\mathcal{W}( \mathcal{D})$ by the relation 
\begin{equation} \label{eq:quasifree}
\tau_t(W(f))=W(T_t f) \quad \mbox{for all } f \in \mathcal{D}.
\end{equation}
It is easy to check that this dynamics, often called quasi-free, is a one-parameter
group of $*$-automorphisms on $\mathcal{W}( \mathcal{D})$.
The goal of this work is to analyze the long time behavior of a
large class of models defined in this manner.

%
%
%
%

\subsection{Harmonic evolutions in infinite volume} \label{subsec:models}
We now introduce a standard class of harmonic systems. The formalism above allows us to define
our models immediately in the infinite volume, i.e. the thermodynamic limit. To
motivate this definition, however, we first recall some well-known, finite volume 
calculations and the requisite notation.

For the models we are interested in $\mathcal{D}$ will be a complex inner product space associated
with $\mathbb{Z}^d$. A common choice is $\mathcal{D} = \ell^2( \mathbb{Z}^d)$, however,  
it is also useful to consider other subspaces,  such as $\mathcal{D} = \ell^1( \mathbb{Z}^d)$
or $\mathcal{D} = \ell^2( \Lambda)$ for some finite $\Lambda \subset \mathbb{Z}^d$. With any of these
choices, the symplectic form $\sigma$ is given by
\begin{equation} \label{eq:hsig}
\sigma(f,g) = \mbox{Im} \left[ \langle f, g \rangle \right] \, \quad \mbox{for } f, g \in \mathcal{D}.
\end{equation}

In finite volume, the models we consider are defined in terms of Hamiltonians
representing a system of coupled harmonic oscillators. For each integer $L \geq 1$,
set $\Lambda_L = \left( -L, L \right]^d  \subset \mathbb{Z}^d$.
The Hamiltonian 
\begin{equation} \label{eq:harmham}
H_L\, = \,  \sum_{ x \in \Lambda_L} p_{ x }^2 \, +\, \omega^2 \, q_{ x}^2 \, + \,
\sum_{j = 1}^{d}  \lambda_j \, (q_{ x } - q_{ x + e_j})^2 \, ,
\end{equation}
equipped with periodic boundary conditions, is a well-defined self-adjoint operator acting  on the Hilbert space
\begin{equation} \label{eq:hspace}
\mathcal{H}_L = \bigotimes_{x \in \Lambda_L} L^2(\mathbb{R}, dq_x).
\end{equation}
To be concrete, the quantities $p_x$ and $q_x$, which appear in (\ref{eq:harmham}) above, are the
single site momentum and position operators regarded as operators on the full Hilbert space $\mathcal{H}_L$
by setting
\begin{equation} \label{eq:pandq}
p_x = \idty \otimes \cdots \otimes \idty
\otimes -i \frac{d}{dq} \otimes \idty \otimes \cdots \otimes \idty \quad
\mbox{ and } \quad q_x = \idty \otimes \cdots \otimes \idty \otimes q \otimes \idty \otimes
\cdots \otimes \idty,
\end{equation}
i.e., these unbounded self-adjoint operators act non-trivially only in the $x$-th factor of $\mathcal{H}_L$. 
It is easy to see that these operators satisfy the canonical commutation relations (CCR), i.e., for all
$x,y \in \Lambda_L$, 
\begin{equation} \label{eq:comm}
[p_x, p_y] \, = \, [q_x, q_y] \, = \, 0 \quad \mbox{ and } \quad
[q_x, p_y] \, = \, i \delta_{x,y} \idty \, .
\end{equation}
With $\{ e_j \}_{j=1}^{d}$, we denote the canonical basis vectors in $\mathbb{Z}^{d}$.
The numbers $\lambda_j \geq 0$ and $\omega > 0$ are the parameters of the system representing
the coupling strength and the on-site energy. As indicated above, $H_L$
is assumed to have periodic boundary conditions, and so we take $q_{x+e_j} = q_{x-(2L-1)e_j}$
if $x \in \Lambda_L$ but $x+ e_j \not\in \Lambda_L$. 

Let us review the finite volume Weyl algebra formalism as it applies in this context.
Set $\mathcal{D}_L = \ell^2( \Lambda_L)$. To each $f \in \mathcal{D}_L$, associate
\begin{equation} \label{eq:hweylop}
W(f) = {\rm exp} \left[ i \sum_{x \in \Lambda_L} \left( {\rm Re}[f(x)] q_x + {\rm Im}[f(x)]p_x \right) \right] \, ,
\end{equation}
a unitary operator on $\mathcal{H}_L$. It is easy to verify that both (\ref{eq:invo}) and (\ref{eq:weylrel}) 
hold for the operators $W(f)$ in (\ref{eq:hweylop}) with $\sigma$ as in (\ref{eq:hsig}); for (\ref{eq:weylrel})
use the Baker-Campbell-Hausdorff relation
\begin{equation}
e^{A+B} = e^Ae^Be^{- \frac{1}{2}[A,B]} \quad \mbox{if} \quad \left[A, \left[ A, B \right] \right] = \left[B, \left[ A, B \right] \right] =0 \, ,
\end{equation}
and the CCR. By construction, each $W(f)$ is unitary, $W(0) = \idty$, and the equality $\| W(f) - \idty \| =2$ follows
since the spectrum of $W(f)$ is the whole of $S^1$. In fact, by (\ref{eq:weylrel}), it is easy to see that
\begin{equation}
W(g) W(f) W(g)^* = e^{i {\rm Im}[ \langle f, g \rangle]} W(f) \, ,
\end{equation}
and so the spectrum is invariant under rotation. Using again Theorem 5.2.8 of \cite{bratteli1997}, 
this proves that, up to $*$-isomorphism, $\mathcal{W}( \mathcal{D}_L)$ is generated by 
Weyl operators as in (\ref{eq:hweylop}).

Since the Hamiltonian $H_L$ in (\ref{eq:harmham}) is self-adjoint, the spectral theorem guarantees that
the Heisenberg dynamics, or time evolution, $\tau_t^L$, given by
\begin{equation}
\tau_t^L(A) = e^{itH_L}A e^{-itH_L} \quad \mbox{for all } t\in \mathbb{R} \quad \mbox{and all } A \in \mathcal{B}( \mathcal{H}_L) \, ,
\end{equation}
is a well-defined, one parameter group of $*$-automorphisms. An important fact is that the harmonic
time evolution $\tau_t^L$ leaves the Weyl algebra $\mathcal{W}( \mathcal{D}_L)$ invariant; this is
proven e.g. in \cite{nachtergaele2010}, by explicitly diagonalizing $H_L$ with Fourier-type operators. In fact, the
formula 
\begin{equation} \label{eq:harminv}
\tau_t^L(W(f)) = W( T_t^L f) \, ,
\end{equation}
is verified in \cite{nachtergaele2010} with mappings $T_t^L$  given by
\begin{equation} \label{eq:deft0}
T_t^L = (U+V) \mathcal{F}^{-1} M_t \mathcal{F} (U^*-V^*) \, .
\end{equation}
Here $\mathcal{F}$ is the unitary Fourier transform on $\mathcal{D}_L$,
$M_t$ is the operator of multiplication by $e^{2i \gamma t}$,  the non-negative function $\gamma$ is
defined by 
\begin{equation} \label{eq:defgamma}
\gamma(k) \, = \,  \sqrt{ \omega^2 \, + \, 4 \sum_{j=1}^{d} \lambda_j \, \sin^2(k_j/2) }
\end{equation}
for all $k \in \Lambda_L^* \, = \, \left\{ \, \frac{x \pi}{L} \, : \, x \in \Lambda_L \, \right\} $,
and $U$ and $V$, known as  Bogoliubov transformations in the literature, see e.g.
\cite{manuceau1968}, are mappings on $\mathcal{D}_L$ given by
\begin{equation} \label{eq:defU+V}
U = \frac{i}{2} \mathcal{F}^{-1} M_{\Gamma_+} \mathcal{F} \quad \mbox{ and } \quad V = \frac{i}{2} \mathcal{F}^{-1} M_{\Gamma_-} \mathcal{F} J
\end{equation}
where $J$ is complex conjugation, and finally $M_{\Gamma_{\pm}}$ is the operator of multiplication
by
\begin{equation} \label{eq:multg}
\Gamma_{\pm}(k) = \frac{1}{\sqrt{ \gamma(k)}} \pm \sqrt{ \gamma(k)} \, ,
\end{equation}
with $\gamma(k)$ as in (\ref{eq:defgamma}). We will not review this calculation here, however, we
will use these results to define a corresponding quasi-free harmonic dynamics on $\mathcal{W}(\ell^2(\mathbb{Z}^d))$.

Let $\mathcal{W}( \ell^2(\mathbb{Z}^d))$ be 
as in Section~\ref{subsec:waqd} with $\sigma$ as in (\ref{eq:hsig}). Regard the function $\gamma$,
previously introduced in (\ref{eq:defgamma}), as a mapping $\gamma : (- \pi, \pi]^d \to \mathbb{R}$. 
Take $\mathcal{F} : \ell^2( \mathbb{Z}^d) \to L^2((-\pi, \pi]^d)$ to be the unitary Fourier transform
and set $U$ and $V$ as in (\ref{eq:defU+V}) with the appropriately extended objects.
Since $\omega >0$, it is clear that both $U$ and $V$ are bounded, real linear transformations on 
$\ell^2( \mathbb{Z}^d)$. One can check that they satisfy
\begin{equation} \label{eq:Bog}
\begin{split}
U^*U-V^*V = \idty = UU^* -VV^* \\
V^*U-U^*V=0 = VU^*-UV^*
\end{split}
\end{equation}
where it is important to note that $V^*$ is the adjoint of the anti-linear mapping $V$.

For each $t \in \mathbb{R}$, we define a mapping $T_t$ on $\ell^2(\mathbb{Z}^d)$ by setting
\begin{equation} \label{eq:deft1}
T_t = (U+V) \mathcal{F}^{-1} M_t \mathcal{F} (U^*-V^*) \, ,
\end{equation}
where, again, $M_t$ is the operator of multiplication on $L^2((-\pi, \pi]^d)$ by $e^{2i \gamma t}$; compare with (\ref{eq:deft0}).
It is easy to see that $T_t$ is a well-defined, real linear mapping. Moreover,  using (\ref{eq:Bog}) one
can readily verify that $\{ T_t \}$ satisfies the group properties  $T_0 = \idty$, $T_{s+t} = T_s \circ T_t$ and,
for each fixed $t$, $T_t$ is symplectic, i.e.,
\begin{equation}
\mbox{Im} \left[ \langle T_t f, T_t g \rangle \right] = \mbox{Im} \left[ \langle f,g \rangle \right] \, .
\end{equation}
As discussed in Section~\ref{subsec:waqd}, Theorem 5.2.8 of \cite{bratteli1997}
demonstrates  in this case the existence of a unique one parameter group of $*$-automorphisms on $\mathcal{W}( \ell^2( \mathbb{Z}^d))$, 
which we will denote by $\tau_t$, that satisfies
\begin{equation}
\tau_t(W(f)) = W(T_tf) \quad \mbox{for all } f \in \ell^2( \mathbb{Z}^d) \, .
\end{equation}
We refer to $\tau_t$ as an infinite volume 
harmonic dynamics on $\mathcal{W}( \ell^2(\mathbb{Z}^d))$.

%
%
%

\subsection{Main Results} \label{sec:mainresults}
In this section, we discuss the two main results of this paper. 
Both begin with the same observation. Using the Weyl relations,
i.e. (\ref{eq:weylrel}), it is easy to see that
\begin{eqnarray}
\left[ \tau_t(W(f)), W(g) \right]  & = & \left\{ W( T_t f) - W(g) W( T_tf) W(-g) \right\} W(g)  \nonumber \\
& = & \left\{ 1 - e^{i {\rm Im}[\langle T_tf, g \rangle]} \right\} W( T_tf) W(g) \, ,
\end{eqnarray}
and by unitarity, this shows that
\begin{equation} \label{eq:easylrb}
\left\| \left[ \tau_t(W(f)), W(g) \right]  \right\| = \left| 1 - e^{i {\rm Im}[\langle T_tf, g \rangle]} \right| \leq \left| \langle T_tf, g \rangle \right|   \, ,
\end{equation}
for all $f,g \in \ell^2( \mathbb{Z}^d)$. A direct calculation, similar to what is done in \cite{nachtergaele2009} for
the finite volume mapping $T_t^L$, shows that
\begin{equation} \label{eq:defft}
T_tf = f * \left(H_t^{(0)} - \frac{i}{2}(H_t^{(-1)} + H_t^{(1)}) \right) + \overline{f}*\left( \frac{i}{2}(H_t^{(1)} - H_t^{(-1)}) \right) \, ,
\end{equation}
where
\begin{equation}\label{eq:h}
\begin{split}
H^{(-1)}_t(x) &= \frac{1}{(2 \pi)^d} {\rm Im} \left[ \int_{(-\pi, \pi]^d} \frac{1}{ \gamma(k)} e^{i(k \cdot x-2\gamma(k)t)} \, d k \right],
\\
H^{(0)}_t(x) &= \frac{1}{(2 \pi)^d}  {\rm Re} \left[ \int_{(-\pi, \pi]^d} e^{i(k \cdot x - 2\gamma(k)t)} \, dk \right],
\\
H^{(1)}_t(x) &=  \frac{1}{(2 \pi)^d} {\rm Im} \left[ \int_{(-\pi, \pi]^d} \gamma(k) \, e^{i(k \cdot x-2\gamma(k)t)} \, dk \right]  \, .
\end{split}
\end{equation}
Combining (\ref{eq:easylrb}) and (\ref{eq:defft}), we find that
\begin{equation} \label{eq:apriest}
\left\| \left[ \tau_t(W(f)), W(g) \right]  \right\|  \leq \sum_{x,y} |f(x)| \, |g(y)| \sum_{m \in \{ -1, 0, 1 \}} |H_t^{(m)}(x-y)| \, .
\end{equation}
All our results follow from this simple estimate.

We can now state our main theorems. Each follows from a careful analysis of the 
behavior of the oscillatory integrals in (\ref{eq:h}). They differ with respect to the
class of allowable Weyl operators. Our first two results apply to all Weyl operators
generated by $f \in \ell^1( \mathbb{Z}^d)$.

\begin{thm} \label{thm:uniest} Let $d \geq 2$ and fix the parameters $\omega >0$ and $\lambda_j >0$ for
all $1 \leq j \leq d$. Denote by $\tau_t$ the harmonic dynamics defined as above on $\mathcal{W}( \ell^2( \mathbb{Z}^d))$. There exists a
number $C>0$, for which, given any $f,g \in \ell^1( \mathbb{Z}^d)$,
the estimate
\begin{equation}	\label{eq:uniestt2}
\left\| \left[ \tau_t(W(f)), W(g) \right]  \right\|  \leq  \min \left[2,  \frac{C \| f \|_1 \| g \|_1}{|t|^{1/2}} \right] \, ,
\end{equation}
holds for all $|t| \geq 1$. 
\end{thm}

As we will see in Section~\ref{sec:uniest}, Theorem~\ref{thm:uniest} follows from 
a bound on the integrals in (\ref{eq:h}) that is uniform with respect to $x \in \mathbb{Z}^d$.
In fact, our proof of Theorem~\ref{thm:uniest} uses that it is not possible for all second order partial derivatives of
$\gamma$ to vanish simultaneously. This is not the case in one dimension.
A slight modification of our argument does provide the following bound.

\begin{thm} \label{thm:uniestt1}
Fix $\omega >0$ and $\lambda >0$ and let $\tau_t$ denote the harmonic dynamics on $\mathcal{W}( \ell^2( \mathbb{Z}))$.
There exists a number $C>0$, for which, given any $f,g \in \ell^1( \mathbb{Z})$,
the estimate
\begin{equation}	\label{eq:uniestt1}
\left\| \left[ \tau_t(W(f)), W(g) \right]  \right\|  \leq \min \left[2, \frac{C \| f \|_1 \| g \|_1}{|t|^{1/3}} \right] \, ,
\end{equation}
holds for all $|t| \geq 1$.
\end{thm}

Theorem~\ref{thm:uniestt1} agrees with a previous result obtained by \cite{marchioro1978} in the special case that $\omega =0$.
For that case, the corresponding integrals are Bessel functions for which the above power-law behavior is sharp.

Our next result demonstrates that, for a restricted class of 
functions $f$ and $g$, one can improve the power of the time decay in (\ref{eq:uniestt2}).
This requires a more detailed analysis of the integrals in (\ref{eq:h}). 
Specifically, one can achieve better time decay with estimates depending on $x \in \mathbb{Z}^d$.
We state this result in terms of Weyl operators generated by
functions in a certain subspace of $\ell^1( \mathbb{Z}^d)$.

Introduce a weight function $w : \mathbb{Z}^d \to [1, \infty)$ by setting
\begin{equation}
w(x) = \left( 1 + \| x \|_1 \right)^{d+3}
\end{equation}
where $\| x \|_1 = \sum_j |x_j|$; the choice of power is a result of our estimates in Section~\ref{sec:mainest}. 
Let us denote by $\ell^1_w( \mathbb{Z}^d)$ the set of
all functions $f: \mathbb{Z}^d \to \mathbb{C}$ for which $fw \in \ell^1( \mathbb{Z}^d)$. 
Since $w \geq 1$, it is clear that  $\ell^1_w(\mathbb{Z}^d)$ is a subspace of $\ell^1(\mathbb{Z}^d)$, and
for any $f \in \ell^1_w(\mathbb{Z}^d)$, denote by
\begin{equation}	\label{eq:l1wnorm}
\|f\|_{1,w} = \sum_{x\in \mathbb{Z}^d} |f(x)| w(x)
\end{equation}
a norm on $\ell^1_w(\mathbb{Z}^d)$. 

We can now state our next result.
\begin{thm} \label{thm:main} Fix the parameters $\omega >0$ and $\lambda_j >0$ for
all $1 \leq j \leq d$ and denote by $\tau_t$ the harmonic dynamics defined as above on $\mathcal{W}( \ell^2( \mathbb{Z}^d))$. There exists a
number $C>0$, for which, given any $f,g \in \ell^1_w(\mathbb{Z}^d)$,
the estimate
\begin{equation} \label{eq:weightbd}
\left\| \left[ \tau_t(W(f)), W(g) \right]  \right\|  \leq \min \left[ 2, \frac{C \| f \|_{1,w} \| g \|_{1,w}}{|t|^{d/2}} \right],
\end{equation}
holds for all $|t| \geq 1$.
\end{thm}

As we show in Section~\ref{sec:mainest}, Theorem~\ref{thm:main} follows from an application of Theorem~\ref{thm:xbd}.
It is interesting to apply the result of Theorem~\ref{thm:xbd} to some particularly simple Weyl operators. For any $x \in \mathbb{Z}^d$,
consider the function $\delta_x: \mathbb{Z}^d \to \mathbb{R}$ given by $\delta_x(y) = 0$ if $y \neq x$ and $\delta_x(x) =1$. A direct
application of Theorem~\ref{thm:xbd} shows that for any $x \in \mathbb{Z}^d$,
\begin{equation}
\left\| \left[ \tau_t(W( \delta_0 )), W( \delta_x) \right]  \right\|  \leq  \frac{C}{|t|^{d/2}} \left(1 + \frac{\|x\|_1^{d+3}}{|t|^{1/2}} \right)
\end{equation}
for all $|t| \geq 1$. In this case, any choice of $x =x(t)$ satisfying 
\begin{equation}
\frac{\| x(t) \|_1^{d+3}}{|t|^{1/2}} = O(1) \quad \mbox{as } |t| \to \infty \, ,
\end{equation}
will have a commutator that decays like $|t|^{-d/2}$. 
Combining this result with the Lieb-Robinson bound, e.g. (\ref{eq:lrb2}), we get the following decay diagram in the $\{ \| x \|_1 ,t\}$ - space:

\begin{figure}[ht]
\begin{picture}(160,160)
\put(-20,10){\vector(1,0){200}}											
\put(185,10){$t$}
\put(80,0){\vector(0,1){140}}
\put(62,143){$\|x\|_1$}
\put(45,125){exp.}
\put(90,125){decay}
\qbezier(180,40)(90,40)(80,10)									
\qbezier(-20,40)(70,40)(80,10)
\put(145,45){$\|x\|_1 = |t|^{1/(2(d+3))}$}
\put(-70,22){$|t|^{-d/2}$-decay}
\put(-40,80){$|t|^{-1/2}$-decay}
\thicklines
\qbezier(80,10)(80,10)(152,127)										
\qbezier(8,127)(80,10)(80,10)
\end{picture}
\caption{Decay of $\left\| \left[ \tau_t(W( \delta_0 )), W( \delta_x) \right]  \right\| $}
\label{fig1}
\end{figure}

\begin{rem}
We expect the decay rate $|t|^{-d/2}$ in the bottom region of the picture to be sharp. However, the rate $|t|^{-1/2}$ as well as the boundary equation $\|x\|_1 = |t|^{1/(2(d+3))}$ are not sharp in all dimensions. In particular, further studies currently in progress show that in dimension two both the rate and the boundary can be improved.    
\end{rem}

As a final remark, we note that although each of the theorems above assume that
$\lambda_j >0$ for all $1 \leq j \leq d$, it is easy to see what happens if some of these are zero.
Fix the dimension $d$ and take parameters $\omega >0$ and
$\lambda_j \geq 0$ for all $1 \leq j \leq d$. For clarity, denote by $H^{(m)}(x; d) = H^{(m)}(x)$ for any $m \in \{ -1, 0, 1\}$
and $x \in \mathbb{Z}^d$ to stress the dimension dependence. Consider the set 
\begin{equation}
A = \{ 1 \leq j \leq d : \lambda_j > 0 \} \, .
\end{equation}
Since $\gamma$ is independent of $k_j$ if $j \notin A$, it is clear that
\begin{equation} \label{eq:hdegen}
H_t^{(m)}(x; d) = H_t^{(m)}( x_A; |A|) \cdot \prod_{j \notin A} \delta_0(x_j) \quad \mbox{for } m \in \{ -1, 0, 1 \} \, .
\end{equation}
Here $x_A$ is the order-preserving, restriction of $x \in \mathbb{Z}^d$ to $\mathbb{Z}^{|A|}$ and $\delta_0 : \mathbb{Z} \to \{ 0, 1\}$ satisfies
$\delta_0(x) = 1$ if and only if $x=0$. Analogues of Theorem~\ref{thm:uniest} and Theorem~\ref{thm:main}, now immediately follow in the
case that some of the couplings are zero.

%
%
%
%

\section{Proof of Theorems \ref{thm:uniest} and \ref{thm:uniestt1}} \label{sec:uniest}

In this section, we will prove Theorems~\ref{thm:uniest} and \ref{thm:uniestt1}.
The section is organized as follows. 
First, we state Theorem~\ref{thm:unibd} below which, for dimensions $d \geq 2$, provides a uniform estimate 
on oscillatory integrals of the type arising in the definition of the harmonic dynamics. 
An immediate consequence of Theorem~\ref{thm:unibd} is Theorem~\ref{thm:uniest}.
Although Theorem~\ref{thm:unibd} does not apply in one dimension, the argument can be
modified to prove a similar bound. We state this result as Theorem~\ref{thm:unibd1d}, and Theorem~\ref{thm:uniestt1}
readily follows.  

We begin with some notation. Fix parameters $\om >0$ and $\la_j > 0$ for $1 \leq j \leq d$. 
For each $t \in \mathbb{R}$ and $x \in \mathbb{Z}^d$, introduce a function $\phi_{t,x} : \mathbb{R}^d \to \mathbb{R}$ by setting
\begin{equation} \label{eq:defphi}
\phi_{t, x}(k) =  k\cdot x  -2 t \gamma(k),
\end{equation}
where, as in \eqref{eq:defgamma},
\begin{equation} \label{eq:defgam2}
 \gamma(k) = \sqrt{\om^2 + 4\sum_{j=1}^d \la_j\sin^2(k_j/2)} \, .
\end{equation}
Here we have written $k_j$ for the $j$-th component of $k \in \mathbb{R}^d$. 
For each $t \in \mathbb{R}$ and $x \in \mathbb{Z}^d$, the function $e^{i \phi_{t,x}}$ is
$2 \pi$-periodic with respect to each component, and so we may regard it 
as a function on the compact torus $\mathbb{T}^d = (- \pi, \pi]^d$.
Let $\| \cdot \|_1$ denote the norm on $L^1(\mathbb{T}^d)$. The following estimate holds.

\begin{thm} \label{thm:unibd}
Fix $d \geq 2$. 
Let  $\eta \in C^1(\mathbb{T}^d)$ and take $\phi_{t, x}$ as defined in \eqref{eq:defphi} above.
There exists a number $C$, independent of $t$ and $x$, for which
\begin{equation} \label{eq:uniest2}
\abs{\int_{\mathbb{T}^d} e^{i\phi_{t, x}(k)}\eta(k)\,dk} \leq \frac{C}{|t|^{1/2}}  \left( \| \nabla \eta \|_1 + \| \eta \|_1 \right) \, ,
\end{equation}
for all $|t| \geq 1$.
\end{thm}

Since the number $C$ is independent of $x \in \mathbb{Z}^d$, Theorem~\ref{thm:uniest} follows from Theorem~\ref{thm:unibd}.

\begin{proof}[Proof of Theorem~\ref{thm:uniest}]
Each Weyl operator $W(f)$ is unitary and thus (\ref{eq:aprioriest}), which also applies in this case, 
demonstrates an upper bound of $2$ for all $t$.
Moreover, using \eqref{eq:h}, it is clear that for any $x \in \mathbb{Z}^d$ and $m \in \{ -1, 0 , 1 \}$,
\begin{equation}
|H_t^{(m)}(x)| \le \frac{1}{(2 \pi)^d}\left| \int_{\mathbb{T}^d} e^{i \phi_{t,x}(k)} \, \gamma^m(k) \,  dk \right| \, .
\end{equation}
In this case, Theorem~\ref{thm:uniest} follows immediately from Theorem~\ref{thm:unibd}.
\end{proof}

We next prove Theorem~\ref{thm:unibd}.

\begin{proof}[Proof of Theorem~\ref{thm:unibd}] 
Our proof of Theorem~\ref{thm:unibd} applies the analysis of Chapter VIII, Sections 1 and 2 of \cite{Stein}
to the oscillatory integral in (\ref{eq:uniest2}) above. Some details are provided to demonstrate that
the prefactors are indeed independent of $x$.

Theorem~\ref{thm:unibd} will follow from an application of Lemma~\ref{lem:prop5} found in Appendix \ref{app:inineq}.
To see that it applies, first observe that
\begin{equation}
 \frac{\partial^2\phi_{t, x}(k)}{\partial k_j\partial k_i} = -2t \frac{\partial^2 \gamma(k)}{\partial k_j\partial k_i} \quad \mbox{for all } 1 \leq i,j \leq d \, .
\end{equation}
We claim that for every $k \in \mathbb{T}^d$, there exists a multi-index $\beta = \beta(k)$ of order 2 with $\partial^\beta \gamma(k) \ne 0.$
In fact, it is easy to see that 
 \begin{equation} \label{eq:gamma2d}
\gamma^3(k) \frac{\partial^2\gamma(k)}{\partial k_j\partial k_i} = \left\{ \begin{array}{cc} \la_j \omega_j \cos(k_j) - \la_j^2 (1- \cos(k_j))^2 & \mbox{if } i = j \, , \\
-\la_j \la_i \sin(k_j) \sin(k_i) & \mbox{otherwise,} \end{array} \right.
\end{equation}
where we have set
\begin{equation}
\om_j = \om^2 + 4\sum_{i\ne j} \la_i\sin^2(k_i/2).
\end{equation}
Hence, if there is a pair $i \neq j$ for which the above mixed derivative is zero, then the second partial derivative with respect to either $k_i$ or
$k_j$ does not vanish. For each $k \in \mathbb{T}^d$, let $\beta(k)$ denote a multi-index of order 2 for which $| \partial^{\beta(k)} \gamma(k)|$ is maximal. 

We now cover the torus with a collection of balls. Fix $k_0 \in \mathbb{T}^d$. It is clear that
for sufficiently small $r>0$,
\begin{equation} \label{eq:gammader}
\abs{ \partial^{\beta(k_0)} \gamma(k_0)} < 2 \abs{ \partial^{\beta(k_0)} \gamma(k)} \quad \mbox{for all } k \in B_r( k_0)  \, ,
\end{equation}
where $B_r(k_0)$ is the ball in $\mathbb{T}^d$ centered at $k_0$ with radius $r$.
Set $r(k_0)$ to be the supremum over all $0 < r \leq 1$ for which (\ref{eq:gammader}) holds.

By compactness, a finite collection of balls $B_{r(k)}(k)$ cover $\mathbb{T}^d$.
Let us index this collection by a finite set $\mathcal{N} \subset \mathbb{T}^d$, i.e. each $n \in \mathcal{N}$
corresponds to a ball $B_{n}$, centered at $n \in \mathbb{T}^d$, and the multi-index $\beta( n)$ is
well-defined.

Set
\begin{equation}
M = \min_{ n \in \mathcal{N}}\inf_{k \in B_{n}} 2 \left| \partial^{\beta(n)} \gamma(k) \right| \, .
\end{equation}
By construction, $M>0$ and obviously, $M$ is independent of $x$.

Let $\{ f_{n} \}_{ n \in \mathcal{N}}$ be a partition of unity subordinate to $\{ B_{ n} \}_{ n \in \mathcal{N}}$, i.e.,
choose $f_{ n} \in C^{\infty}(\mathbb{T}^d)$ with $0 \leq f_{ n}(k) \leq 1$, ${\rm supp}(f_{ n}) \subset B_{n}$, and 
\begin{equation}
\sum_{n \in \mathcal{N}} f_{ n}(k) = 1 \quad \mbox{for all} \quad k \in \mathbb{T}^d \, .
\end{equation}
It is clear then that
\begin{equation} \label{eq:partitionineq1}
\left| \int_{\mathbb{T}^d} e^{i\phi_{t, x}(k)}\eta(k)\,dk \right| \leq \sum_{ n \in \mathcal{N}} \left| \int_{\mathbb{T}^d} e^{i\phi_{t, x}(k)}\eta_{ n}(k)\,dk\right|
\end{equation}
where we have set $\eta_{n} = f_{ n} \eta$. 

Since each $f_n$, and thereby each $\eta_n$, has support on a ball with radius less than one,
it is convenient to regard the integrands above as functions on $\mathbb{R}^d$ with support on
a connected ball whose radius is also less than one. Given this, each of the integrals 
above can be estimated using Lemma~\ref{lem:prop5} in the Appendix. 

To see this, fix $n \in \mathcal{N}$. Identify $n \in \mathbb{T}^d$ with it's representative in $(- \pi, \pi]^d$. 
If, upon identification, the ball $B_n \subset (-\pi, \pi]^d$ as a subset of $\mathbb{R}^d$, then
extend $\eta_n$ to $\mathbb{R}^d \setminus B_n$ by zero. In this case, we write
\begin{equation}
\int_{\mathbb{T}^d} e^{i \phi_{t, x}(k)} \eta_n(k) \, dk = \int_{\mathbb{R}^d} e^{i \phi_{t, x}(k)} \eta_n(k) \, dk  \, ,
\end{equation}
with a slight abuse of notation.

Otherwise, under identification, the ball $B_n$ is not a subset of $(- \pi, \pi]^d$.
In this case, denote by $\tilde{\eta}_n$ the periodic extension of $\eta_n$ to $\mathbb{R}^d$.
Let $\tilde{B}_n$ denote the ball in $\mathbb{R}^d$ centered at $n \in (- \pi, \pi]^d$ with radius equal
to that of $B_n$. It is clear that
\begin{equation}
\int_{\mathbb{T}^d} e^{i \phi_{t, x}(k)} \eta_n(k) \, dk = \int_{\tilde{B}_n} e^{i \phi_{t, x}(k)} \tilde{\eta}_n(k) \, dk  \, .
\end{equation}
With another slight abuse of notation, we will re-designate $\eta_n$ to be the
function on $\mathbb{R}^d$ that extends $\tilde{\eta}_n |_{\tilde{B}_n}$ to $\mathbb{R}^d$
by zero and declare its support to be $B_n$. 

With this understanding, for each $n \in \mathcal{N}$, we have that ${\rm supp}( \eta_{n}) \subset B_{n}$ 
and moreover, 
\begin{equation}
\abs{\partial^{\beta( n)}\phi_{t, x}(k)} \geq M|t| \quad \text{ for all } k \in B_{n} \, .
\end{equation}
Clearly then, Lemma~\ref{lem:prop5} applies and demonstrates that
\begin{eqnarray} 
\left| \int_{ \mathbb{T}^d} e^{i\phi_{t, x}(k)}\eta(k)\,dk \right| & \leq & \sum_{n \in \mathcal{N}} \frac{C_n}{ \sqrt{|t|} } \left(   \int_{\mathbb{R}^d} \abs{ \nabla \eta_{n}(k) } \, dk  + 
\int_{\mathbb{R}^d} \abs{\eta_{n}(k)} \, dk  \right) \nonumber \\ 
& = & \frac{C}{ \sqrt{|t|}}  \left( \left\| \nabla \eta \right\|_1 + \| \eta \|_1 \right) ,
\end{eqnarray}
proving \eqref{eq:uniest2}. 
\end{proof}

In one dimension, it is possible for the second derivative of $\gamma$ to vanish in $(- \pi, \pi]$, and thus, the proof of
Theorem~\ref{thm:unibd} does not apply. There is, however, a related bound.

\begin{thm} \label{thm:unibd1d}
Let  $\eta \in C^1(\mathbb{T}^1)$ and take $\phi_{t, x}$ be as defined in \eqref{eq:defphi} above.
There exists a number $C$, independent of $t$ and $x$, for which 
\begin{equation} \label{eq:uniest1d}
\abs{\int_{\mathbb{T}^1}e^{i\phi_{t, x}(k)}\eta(k)\,dk} \leq \frac{C}{|t|^{1/3}} \left( \| \eta' \|_1 + \| \eta \|_1 \right)
\end{equation}
whenever $|t| \geq 1$. 
\end{thm}

It is clear from the argument immediately following Theorem~\ref{thm:unibd} that Theorem~\ref{thm:uniestt1} follows from
Theorem~\ref{thm:unibd1d}. We now prove Theorem~\ref{thm:unibd1d}.

\begin{proof}[Proof of Theorem~\ref{thm:unibd1d}] 
We begin by observing that in one dimension, the second and third derivatives of $\phi_{x,t}$ 
cannot simultaneously vanish. Although it is also true that, for sufficiently
large $t$, the first and second derivatives cannot both be zero, the dependence of the first on
$x$ prevents uniform estimates.

A short calculation shows that
\begin{equation} \label{eq:phi2d}
\frac{\partial^2\phi_{t, x}(k)}{\partial k^2} = - 2\la^2 t \frac {\frac{\om^2}{\la}\cos(k) - (1- \cos(k))^2}{\gamma^3(k)}
\end{equation}
and
\begin{equation} \label{eq:phi3d}
\frac{\partial^3\phi_{t, x}(k)}{\partial k^3}  = \frac {2\la\sin(k) t}{\gamma(k)}\left(1 + 3\la^2\frac {\frac{\om^2}{\la}\cos(k) - (1 - \cos(k))^2}{\gamma^4(k)}\right).
\end{equation}
Using (\ref{eq:phi2d}) and (\ref{eq:phi3d}) above, it is clear that if the third derivative
vanishes, then the second is non-zero. Moreover, neither of these derivatives 
depend on $x$. 

Fix $k_0 \in \mathbb{T}^1$ and denote by $\gamma^{(i)}$ the $i$-th derivative of $\gamma$. 
{F}rom the above, it is clear that at least one of $\gamma^{(2)}(k_0)$ and $\gamma^{(3)}(k_0)$ is non-zero. 
For $i =2,3$, define $r_i(k_0) = 0$ if $\gamma^{(i)}(k_0) = 0$. Otherwise, there exists $r_i>0$ for which
\begin{equation} \label{eq:gammider}
\abs{ \gamma^{(i)}(k)} > 0 \quad \mbox{for all } k \in B_{r_i}(k_0) \ .
\end{equation}
In this case, set $r_i(k_0)$ to be the supremum over all $0 < r_i \leq 1$ such that (\ref{eq:gammider}) holds and
take $2r(k_0) = \max[r_2(k_0), r_3(k_0)]$. By construction, it is clear that $r(k) >0$ for all
$k \in \mathbb{T}^1$.

We now proceed as in the proof of Theorem~\ref{thm:unibd}. By compactness, a finite collection of balls with the form $B_{r(k)}(k)$ cover $\mathbb{T}^1 $.
Let us denote this collection by $\{ B_{n} \}_{n \in \mathcal{N}}$.
It is convenient to denote by
\begin{equation}
\mathcal{N}_2 = \left\{ n \in \mathcal{N} : |\gamma^{(2)}(k) | >0 \, \mbox{ for all } k \in \overline{B_{n}} \right\} \, ,
\end{equation}
and $\mathcal{N}_3 = \mathcal{N} \setminus \mathcal{N}_2$. With
\begin{equation}
B^i = \bigcup_{ n \in \mathcal{N}_i} \overline{B_{n}} \, \quad \mbox{ for } i = 2, 3, 
\end{equation}
the quantities
\begin{equation}
M_i = \inf_{k \in B^i} 2 \left|\gamma^{(i)}(k) \right| \, ,
\end{equation}
are strictly positive and independent of $x \in \mathbb{Z}$.

Let $\{ f_{n} \}_{n \in \mathcal{N}}$ be a partition of unity subordinate to $\{ B_{n} \}_{n \in \mathcal{N}}$, and note
that 
\begin{equation} \label{eq:partitionineq}
\left| \int_{ \mathbb{T}^1} e^{i\phi_{t, x}(k)}\eta(k)\,dk \right| \leq \sum_{n \in \mathcal{N}} \left| \int_{\mathbb{T}^1} e^{i\phi_{t, x}(k)}\eta_{n}(k)\,dk\right|
\end{equation}
where we have set $\eta_{n} = f_{n} \eta$. As before, see proof of Theorem~\ref{thm:unibd}, it is convenient to regard the integrands above as
functions supported on connected balls in $\mathbb{R}$.

There are two cases. First, suppose $n \in \mathcal{N}_2$. In this case, it is clear that
\begin{equation}
\abs{\frac{\partial^2}{ \partial k^2} \phi_{t,x}(k)} = \abs{2 t \gamma^{(2)}(k)} \geq |t| M_2 > 0 \quad \mbox{for all} \quad k \in B_{n} \, .
\end{equation}
Since ${\rm supp}( \eta_{n}) \subset B_{n}$,
Lemma~\ref{lem:1dest} applies and 
\begin{eqnarray}
\left| \int_{B_n} e^{i\phi_{t, x}(k)}\eta_{n}(k)\,dk\right| & \leq & \frac{c_2}{\sqrt{|t| M_2}} \int_{B_n} \abs{ \eta'_{n}(k)} \, dk  \, \nonumber \\
& \leq & \frac{C}{ \sqrt{|t|}} \left( \| \eta \|_1 + \| \eta' \|_1 \right) \, .
\end{eqnarray}

Now for $n \in \mathcal{N}_3$, we have that
\begin{equation}
\abs{\frac{\partial^3}{ \partial k^3} \phi_{t,x}(k)} = \abs{2 t \gamma^{(3)}(k)} \geq |t| M_3 > 0 \quad \mbox{for all} \quad k \in B_{n} \, .
\end{equation}

Another application of Lemma~\ref{lem:1dest} shows that 
\begin{equation}
\left| \int_{B_n} e^{i\phi_{t, x}(k)}\eta_{n}(k)\,dk\right| \leq \frac{C'}{\sqrt[3]{|t|}} \left( \|\eta \|_1 + \| \eta' \|_1 \right) \, .
\end{equation}
The bound (\ref{eq:uniest1d}) readily follows for $|t| \geq 1$.
\end{proof}

%
%
%
%

\section{Proof of Theorem \ref{thm:main} } \label{sec:main}

Our proof of Theorem~\ref{thm:main} uses a general estimate on oscillatory integrals whose
phase is given by the function $\phi_{t,x}$ defined in (\ref{eq:defphi}).
In the beginning of Section~\ref{sec:mainest} we state this result, see Theorem~\ref{thm:xbd}, and
show how Theorem~\ref{thm:main} follows. 
The remainder of Section~\ref{sec:mainest} contains the
proof of Theorem~\ref{thm:xbd}. One part of the proof uses the Morse
Lemma.
For the sake of completeness, we briefly discuss this application in Section~\ref{sec:Morse}.

\subsection{The Main Estimates} \label{sec:mainest}

Before stating Theorem~\ref{thm:xbd}, we introduce some notation which will be
used throughout this section. Let $f: \mathbb{R}^d \to \mathbb{C}$ be smooth. For any
$1 \leq j \leq d$ and any $N \geq 1$ set
\begin{equation} \label{eq:fnorm}
\| f \|_{j, N} = \max_{0 \leq k \leq N} \| \partial_j^{(k)} f \|_{\infty} \, ,
\end{equation}
where we have denoted by $\partial_j^{(k)}$ the $k$-th partial derivative of $f$ 
with respect to the $j$-th component. Recall that for each $t \in \mathbb{R}$ and
$x \in \mathbb{Z}^d$ we have defined functions $\phi_{t,x}$ and $\gamma$ as in
(\ref{eq:defphi}) and (\ref{eq:defgam2}) respectively.

\begin{thm} \label{thm:xbd} 
Let $\eta \in C^{d+3}(\mathbb{T}^d)$ and take $\phi_{t, x}$ as defined in \eqref{eq:defphi}. 
There exists a number $C$, independent of $t$ and $x$, for which
\begin{equation} \label{eq:xintbd}
\abs{\int_{ \mathbb{T}^d} e^{i\phi_{t, x}(k)}\eta(k)\,dk} \leq  C \left( \frac{ \| \eta \|_{\infty}}{|t|^{d/2}} + \sup_{\beta : | \beta | \leq d+3} \| \partial^{\beta} \eta \|_{\infty} \cdot \frac{(3+ \| x \|_1)^{d+3}}{|t|^{(d+1)/2}} \right) 
\end{equation}
for all $|t| \geq 1$.
\end{thm}

Here for each $x \in \mathbb{Z}^d$ we have denoted by $\| x \|_1 = \sum_{j=1}^d | x_j|$, the norm on $\eta$ and $\partial^{\beta} \eta$ corresponds to $L^{\infty}( \mathbb{T}^d)$, and the supremum is over all multi-indices $\beta$ whose order satisfies $| \beta| \leq d+3$.
We now show that Theorem~\ref{thm:main} follows from Theorem~\ref{thm:xbd}.

\begin{proof}[Proof of Theorem~\ref{thm:main}]
We argue as in the proof of Theorem~\ref{thm:uniest}. Using (\ref{eq:apriest}) and Theorem~\ref{thm:xbd}, we conclude
that
\begin{equation}
\left\| \left[ \tau_t(W(f)), W(g) \right]  \right\|  \leq  \frac{C}{|t|^{d/2}}\sum_{x,y} |f(x)| \, |g(y)| \left(1 + (3+\|x-y\|_1)^{d+3} \right)
\end{equation}
for all $|t| \geq 1$. Since
\begin{equation}
(3+\|x-y\|_1)^{d+3} \le C' \left(1+\|x\|_1\right)^{d+3}\left(1+\|y\|_1\right)^{d+3},
\end{equation}
it is clear that (\ref{eq:weightbd}) follows.
\end{proof}

\begin{proof}[Proof of Theorem~\ref{thm:xbd}]
As in Section~\ref{sec:uniest}, we apply the analysis of Chapter VIII, Sections 1 and 2 of
\cite{Stein} to the oscillatory integral in (\ref{eq:xintbd}).

For large $|t|$, the leading order contribution in the estimate (\ref{eq:xintbd}) 
will be determined by the critical points of $\gamma$, i.e., 
those $k \in (- \pi, \pi]^d$ for which $\nabla \gamma (k) = 0$. 
Since
\begin{equation}
\frac{ \partial \gamma}{ \partial k_j} (k) = \lambda_j \frac{ \sin(k_j)}{ \gamma(k)} \quad \mbox{for all } 1 \leq j \leq d \, ,
\end{equation}
it is clear that there are $2^d$ such critical points; 
namely those $k$ with $k_j \in \{0, \pi \}$ for all $1 \leq j \leq d$. 
Let us denote the set of critical points of $\gamma$ as $\Upsilon$.

We now partition the torus $\mathbb{T}^d$ so that these critical points lie in distinct balls.
In Section~\ref{sec:Morse}, we show that each of these critical points is non-degenerate, and
therefore, the Morse Lemma applies. Fix $0<r \leq 1$ as described in Section~\ref{sec:Morse} and 
for each $k \in \Upsilon$, denote by $B_r(k) \subset \mathbb{T}^d$ the open ball of radius $r$ centered at $k$. 
Set
\begin{equation}
B_r(\Upsilon) = \bigcup_{k \in \Upsilon} B_r(k) \, ,
\end{equation}
and observe that for each $k \in \mathbb{T}^d \setminus B_{r}(\Upsilon)$, there exists a number $1 \leq j \leq d$ such that 
\begin{equation} \label{eq:radlwbd}
\min_{m\in\{-1, 0, 1\}}\abs{k_j-m\pi} \ge r/\sqrt{d}.
\end{equation}
In this case, there exists a number $\delta_k>0$ such that 
\begin{equation}\label{eq:compineq}
\min_{m\in\{-1, 0, 1\}}\abs{k'_j-m\pi} \ge r/(2\sqrt{d}) \quad \mbox{for all } k' \in B_{\delta_k}(k) \, .
\end{equation}
Set $\delta(k)$ to be the supremum over all $0< \delta_k \leq 1$ such that given (\ref{eq:radlwbd}), (\ref{eq:compineq}) holds.
By compactness, the set $\mathbb{T}^d\setminus B_{r}(\Upsilon)$
can be covered by finitely many such balls. 
This construction covers the torus by a finite union of balls;
label them as $\{ B_n \}_{n \in \mathcal{N}}$.

Let $\{ f_n \}_{n \in \mathcal{N}}$ be a partition of unity subordinate to this cover and estimate
\begin{align*}
\abs{\int_{\mathbb{T}^d} e^{i \phi_{t,x}(k)} \, \eta(k) \, dk} & \leq
\sum_{n \in\mathcal{N}} \abs{\int_{\mathbb{T}^d} e^{i \phi_{t,x}(k)}
\, \eta_n( k) dk } \, ,
\end{align*}
where we again set $\eta_n = f_n \eta$. For the remainder of this argument, we will, as in the proof of Theorem~\ref{thm:unibd}, regard the
integrands above as functions on $\mathbb{R}^d$ supported on a connected ball with radius less than one. 

To prove (\ref{eq:xintbd}), we need only estimate the integrals above
for each $n \in \mathcal{N}$; here we are using that the cardinality of 
$\mathcal{N}$ depends only on the function $\gamma$ and, in particular, it is independent of $x$ and $t$. There are two cases. 
First, let $n \in \mathcal{N}$ correspond to one of the balls covering $\mathbb{T}^d \setminus B_r( \Upsilon)$, i.e.
a function $\eta_n$ whose support does not include a critical point of $\gamma$. In this case,
by (\ref{eq:compineq}) above, there exists a coordinate $k_j$ and a number $\alpha >0$ for which
\begin{equation}
\abs{\frac{ \partial \gamma}{ \partial k_j} (k')} =  \lambda_j \abs{\frac{ \sin(k_j')}{ \gamma(k')}} \geq \alpha \quad \mbox{for all } k' \in B_n \, .
\end{equation}
Here $\alpha$ can be made explicit in terms of $\omega$, $\{ \lambda_j \}$, $r>0$, and $d$, but most importantly, it
is independent of both $t$ and $x$. An application of Lemma~\ref{l:outside_critical2} shows that
\begin{equation}
\abs{ \int_{\mathbb{R}^d} e^{i \phi_{t,x}(k)} \eta_n(k) \,dk } = \abs{ \int_{\mathbb{R}^d} e^{- i 2t \gamma(k)} e^{ikx} \eta_n(k) \,dk }  \leq \frac{ C_N(\gamma, e^{i \cdot x} \eta_n( \cdot) )}{|t|^N} \, ,
\end{equation}
for any $N \geq 1$ and $t \neq 0$. As the proof of Lemma~\ref{l:outside_critical2} demonstrates,
\begin{eqnarray}
C_N(\gamma, e^{i \cdot x} \eta_n( \cdot) ) & = & \sum_{\ell =0}^N \int_{\mathbb{R}^d} \alpha^{\ell-2N} \abs{ \partial_j^{(\ell)} e^{ikx} \eta_n(k)} \abs{P_{N,\ell}(k)} \, dk \nonumber \\
& \leq & \frac{1}{ \alpha^{2N}} \sum_{\ell =0}^N \alpha^{\ell} \sum_{m=0}^{\ell} {\ell \choose m} |x_j|^m \int_{B_n} \abs{ \partial_j^{(\ell-m)} \eta_n(k)} \abs{P_{N,\ell}(k)} \, dk \nonumber \\
& \leq &  C_{N}(\gamma, n) \| \eta \|_{j, N} (1 + \| x \|_{\infty})^{N} \, .
\end{eqnarray}
Note that quantity $\| \eta \|_{j,N}$ in the final line above is as in (\ref{eq:fnorm}) with norms in $L^{\infty}( \mathbb{T}^d)$. 
Moreover, $\| x \|_{\infty} \leq \| x \|_1$, and for any $d \geq 1$, $d+1 \leq 2(d+3)$. Thus, choosing $N = d+3$, we get
an estimate of the form (\ref{eq:xintbd}) for all such $n$.

Next, we consider those $n \in \mathcal{N}$ for which the support of $\eta_{n}$ contains a critical
point. Denote by $k_n$ the unique element of $\mbox{supp}(\eta_n) \cap \Upsilon$. By construction, namely
the initial choice of $r>0$, $\gamma$ satisfies a Morse Lemma at  $k_n$ throughout the ball $B_n$.  

Translating to $k_n$, it is clear that
\begin{eqnarray} \label{eq:inttrans}
\abs{ \int_{\R^d} e^{i \phi_{t,x}(k)} \eta_n(k) \, dk} & = & \abs{ \int_{\R^d} e^{-2it \left( \gamma(k+k_n) - \gamma(k_n) \right)} e^{i k \cdot x} \, \eta_n(k + k_n) \,dk} \, \nonumber \\
& = & \abs{ \int_{\R^d} e^{-2it Q(y)} e^{i x\cdot \mathcal{D}^{-1}(y)} \, \eta_n(\mathcal{D}^{-1}(y) + k_n)\abs{\det( \nabla \mathcal{D}^{-1}(y))} dy} \, \nonumber \\
& = & \abs{ \int_{\R^d} e^{-2it Q(y)} e^{-|y|^2} G_n(y;x) \, dy} \, ,
\end{eqnarray}
where, for the second equality above, we changed variables $y = \mathcal{D}(k)$ according to the Morse Lemma, see Section~\ref{sec:Morse}, and
denoted the resulting quadratic function by 
\begin{equation}
Q(y) = \sum_{j=1}^d \epsilon_j y_j^2 \, \quad \mbox{with } \epsilon_j \in \{ \pm 1 \} \, .
\end{equation}
For the final equality in (\ref{eq:inttrans}) above, we inserted and removed a gaussian, and set 
\begin{equation} \label{eq:defg}
G_n(y;x) = e^{|y|^2} e^{i x\cdot \mathcal{D}^{-1}(y)} \, \eta_n(\mathcal{D}^{-1}(y) + k_n)\abs{\det( \nabla \mathcal{D}^{-1}(y))}  \, .
\end{equation} 
Since $\mathcal{D}$ is a diffeomorphism, the determinant above does not change sign. 
Without loss of generality, we assume it is positive for the calculations below.
Moreover, we recall that $\mathcal{D}(0) = 0$.

We need only estimate the integral on the right-hand side of (\ref{eq:inttrans}) above.
To do so, we re-write it as a sum of three.
Fix $h \in C_0^{\infty}( \mathbb{R}^d)$ satisfying
$h(y) = 1$ for all  $y \in \mbox{supp}(G_n) \cup B_1(0)$.
{F}rom \eqref{eq:inttrans}, it is clear that 
\begin{equation}			\label{eq:I123}
\abs{ \int_{\R^d} e^{i \phi_{t,x}(k)} \eta_n(k) \, dk} \leq I_1(t) + I_2(t) + I_3(t;x),
\end{equation}
where
\begin{eqnarray}
I_1(t) &= & \abs{G_{n}(0;x)} \abs{\int_{\R^d} e^{-2 i t Q(y)} e^{-\abs{y}^2} dy}, \nonumber \\
I_2(t) &= & \abs{G_{n}(0;x)} \abs{\int_{\R^d} e^{-2 i t Q(y)} e^{-\abs{y}^2} (h(y) - 1)dy},   \\
I_3(t;x) & = & \abs{ \int_{\R^d} e^{-2 i t Q(y)} e^{-\abs{y}^2} \left(G_{n}(y;x) -G_{n}(0;x) \right) h(y) dy} \, . \nonumber
\end{eqnarray}
The proof of Theorem~\ref{thm:xbd} is complete when we show that each of the
integrals above satisfy an estimate of the form (\ref{eq:xintbd}).

$I_1$ and $I_2$ are easy to bound. In fact, since $\mathcal{D}(0) = 0$, (\ref{eq:defg}) implies that
\begin{equation}
\abs{G_n(0;x)} = \frac{ | \eta_n(k_n)|  }{ \abs{\det( \nabla \mathcal{D}(0))}} \, ,
\end{equation}
and thus, as the notation implies, both $I_1$ and $I_2$ are independent of $x$. In this case,
an application of Lemma~\ref{lem:Qest} shows that 
\begin{equation} \label{eq:I1bd}
I_1(t) \le  \frac{ \| \eta \|_{\infty} }{ \abs{\det( \nabla \mathcal{D}(0))}}  \left( \frac{\pi}{2|t|} \right)^{d/2} \quad \mbox{for all  } t \neq 0 \, ,
\end{equation}
while, for any integer $N \geq 1$, Lemma~\ref{lem:Qetaest} yields a bound of the form
\begin{equation}	\label{eq:I2bd}
I_2(t) \le \frac{ \| \eta \|_{\infty} }{ \abs{\det( \nabla \mathcal{D}(0))}}   \frac {C_N}{|t|^N} \, ,
\end{equation}
for all $t \neq 0$. The bound on $I_1$ produces the first term in (\ref{eq:xintbd}) and
by taking $N=d+3$, it is clear that $I_2$ also satisfies (\ref{eq:xintbd}). 

We need only estimate the integral $I_3$, and for this we use Lemma~\ref{lem:gen}. 
To see that it is applicable, we Taylor expand $G_{n}(y;x)$. Write
\begin{eqnarray}  \label{taylor}
G_{n}(y;x) & = &  G_{n}(0; x) + \sum_{j = 1}^d \int_0^{y_j} \partial_j G_{n}(y_1,.., y_{j-1}, s, 0, ..,0;x) ds \nonumber \\
& = & G_{n}(0; x) + \sum_{j = 1}^d y_j \tilde{G}_{n, j}(y;x) \, ,
\end{eqnarray}
where we have written $\partial_j G_{n}$ for the partial derivative of $G_{n}$ with respect to $y_j$ and set
\begin{equation} \label{eq:deftg}
\tilde{G}_{n, j}(y;x) = \frac {1}{y_j} \int_0^{y_j} \partial_j G_{n} (y_1,.., y_{j-1}, s, 0, ..,0;x) ds
\end{equation}
if $y_j \neq 0$, and $\partial_jG_{n}(y_1, ..., y_{j-1}, 0, ..., 0;x)$ otherwise. Clearly, then
\begin{eqnarray}		\label{eq:I3bd0}
I_3(t;x) & \leq & \sum_{j=1}^d \abs{\int_{\R^d} e^{-2 i t Q(y)} e^{-\abs{y}^2} y_j \tilde{G}_{n, j}(y;x) h(y) dy} \nonumber \\
& \leq & \frac{C}{|t|^{(d+1)/2}} \sum_{j=1}^d \left( \| \tilde{G}_{n,j} h \|_{\infty} + \sum_{k=1}^d \| \mathcal{G}_{j,k} \|_{k, d+2} \right) \, ,
\end{eqnarray}
for all $|t| \geq 1$. Here we have set $\mathcal{G}_{j,k}(y) = e^{- y_k^2} \tilde{G}_{n,j}(y) h(y)$ 
as in Lemma~\ref{lem:gen}. The quantities on the right-hand side of (\ref{eq:I3bd0})  depend on $x$. We will now prove
that there exists a number $C'$, independent of $x$ and $t$, for which
\begin{equation} \label{eq:I3bd}
I_3(t,x) \leq C' \sup_{\beta : | \beta| = d+3} \| \partial^{\beta} \eta \|_{\infty} \frac{(3+ \| x \|_1)^{d+3}}{|t|^{(d+1)/2}} \, .
\end{equation}
Given (\ref{eq:I3bd}), the proof of Theorem~\ref{thm:xbd} is complete.

The proof of (\ref{eq:I3bd}) will be performed in two steps. First, we estimate the partial derivatives of $G_{n}$ in terms of $x$ and $\eta$. Second, we estimate partial derivatives of $\mathcal{G}_{j,k}$ in terms of the partial derivatives of $G_{n}$. The combination of the two steps will give us the necessary bound.

Let us write
\begin{equation}
G_{n}(y;x) = g_1(y;x) g_2(y) \, 
\end{equation}
where
\begin{equation}
g_1(y;x) = e^{i x \cdot \mathcal{D}^{-1}(y)} \eta( \mathcal{D}^{-1}(y) + k_n) \quad \mbox{and} \quad g_2(y) = f_{n}(\mathcal{D}^{-1}(y) + k_n) e^{|y|^2} {\rm det}( \nabla \mathcal{D}^{-1}(y)) \, .
\end{equation}
The derivatives we must bound are of the form
\begin{equation}
\partial_k^{(\ell)} \partial_j G_{n}(y;x) = \sum_{m=0}^{\ell} {\ell \choose m} \left[ \partial_k^{(m)} \partial_j g_1(y;x) \partial_k^{(\ell -m)} g_2(y) +  \partial_k^{(m)} g_1(y;x) \partial_k^{(\ell -m)} \partial_j g_2(y) \right] \, .
\end{equation}
A short calculation shows then that
\begin{eqnarray}
\left\| \partial_k^{(\ell)} \partial_j G_{n} \right\|_{\infty} & \leq &  C_{k,j}( \ell)  \sum_{m=0}^{\ell} {\ell \choose m} \left[ \left\| \chi_n \partial_k^{(m)} \partial_j g_1\right\|_{\infty} +  \left\| \chi_n \partial_k^{(m)} g_1\right\|_{\infty} \right]  \nonumber \\
& \leq & C'_{k,j}( \ell) (2+ \| x \|_1)^{\ell +1}  \left( \| \eta \|_{k, \ell} + \| \partial_j \eta \|_{k, \ell} \right)  \, .
\end{eqnarray}
In the intermediate step the notation $\chi_n$ stands for the characteristic function of the ball $B_n$.

We can now estimate the right-hand side of (\ref{eq:I3bd0}) term by term. Clearly, 
\begin{equation}
\| \tilde{G}_{n, j} h \|_{\infty} \leq \| \partial_j G_{n} \|_{\infty} \| h \|_{\infty} \leq C_j \| \eta \|_{j,1} \left(1 + \| x \|_1 \right) \, .
\end{equation}

The bound on $\mathcal{G}_{j, k}$ depends on $j$ and $k$.
In fact, by construction one sees that $\tilde{G}_{n, j}$ depends only on those coordinates $k$
with $1 \leq k \leq j$. In particular,
\begin{equation}
\partial_k^{(\ell)} \tilde{G}_{n, j}(y;x) = 0 \quad \mbox{for any} \quad j < k \leq d \quad \mbox{and} \quad \ell \geq 1 \, .
\end{equation}
Thus if $j < k \leq d$,
\begin{equation}
\| \mathcal{G}_{j,k} \|_{k, d+2} \leq \| \tilde{G}_{n, j} \|_{\infty} \| e^{- y_k^2} h \|_{k, d+2} \leq C_k \| \eta \|_{j,1} \left(1 + \| x \|_1 \right)  \, .
\end{equation}

If, on the other hand, $1 \leq k <j$, then
\begin{equation}
\left\| \partial_k^{(\ell)}\tilde{G}_{n, j} \right\|_{\infty} \leq  \left\| \partial_k^{(\ell)} \partial_j G_{n} \right\|_{\infty} \, ,
\end{equation}
and so
\begin{eqnarray} \label{eq:tgbd}
\| \mathcal{G}_{j,k} \|_{k, d+2}  & \leq & \max_{0 \leq \ell \leq d+2}  \sum_{m=0}^{\ell} { \ell \choose m} C_{k, \ell}(m) \| \partial_k^{(m)} \tilde{G}_{n, j} \|_{\infty} \nonumber \\
& \leq & C \left(  \| \eta \|_{k, d+2} + \| \partial_j \eta \|_{k, d+2} \right) \left( 3 + \| x \|_1 \right)^{d+3} \, .
\end{eqnarray}

For the case of $k=j$, observe that for every $\ell \geq 1$,
\begin{eqnarray}
\tilde{G}_{n, j}(y;x) & = & \sum_{m=1}^{\ell} \frac{y_j^{m-1}}{m!} \partial_j^{(m)}G_{n}(y_1, ..., y_{j-1}, 0, ..., 0;x) \nonumber \\
& \mbox{ } & \quad + \frac{1}{y_j} \int_0^{y_j} \int_0^{s_1} \cdots \int_0^{s_{\ell}} \partial_j^{(\ell+1)}G_{n}(y_1, ..., y_{j-1}, s_{\ell +1}, 0, ..., 0;x)  \prod_{p=1}^{\ell +1} d s_p \, .
\end{eqnarray}
As a result, it is clear that
\begin{equation}
\left\| \partial_j^{(\ell)} \tilde{G}_{n, j} \right\|_{\infty} \leq C_j( \ell) \left\| \partial_j^{(\ell+1)} G_{n} \right\|_{\infty} \, ,
\end{equation}
and therefore, 
\begin{equation}
\| \mathcal{G}_{j,j} \|_{j, d+2} \leq C_j \| \eta \|_{j, d+3} \left( 3 + \| x \|_1 \right)^{d+3} \, .
\end{equation}
This proves (\ref{eq:I3bd}) and hence Theorem~\ref{thm:xbd}.
\end{proof}
%
%
%
%

\subsection{On the Morse Lemma for $\gamma$} \label{sec:Morse}

The Morse Lemma demonstrates the existence of a convenient change of variables for 
functions in a neighborhood of a non-degenerate critical point. Let $f: \mathbb{R}^d \to \mathbb{R}$
have a critical point at $x_0$, i.e., $\nabla f(x_0) = 0$. The critical point $x_0$ is said to be 
non-degenerate if the matrix of second partial derivatives of $f$ is
invertible at $x_0$. Given a function $f$ with a non-degenerate critical point $x_0$, the Morse Lemma
shows that there exists a neighborhood $U$ of $x_0$ in which
\begin{equation} \label{eq:cov}
f(x) = f(x_0) + \sum_{j=1}^d \epsilon_j y_j(x)^2 \quad \mbox{for all } x \in U \, 
\end{equation}
and some numbers $ \epsilon_j \in \{ \pm 1 \}$. More precisely, the Morse Lemma proves that there exists a neighborhood $U$ of $x_0$ and
a diffeomorphism $\mathcal{D}$ from $U$ into a neighborhood of the origin such that (\ref{eq:cov}) holds
with $y_j(x)$ representing the $j$-component of $\mathcal{D}(x)$. Clearly then, $\mathcal{D}(x_0) =0$. Since the Morse Lemma is constructive, it
is particularly useful in calculations. We refer the interested reader to Chapter VIII, Section 2.3.2 of \cite{Stein} for a constructive proof of the Morse Lemma.

Consider the function $\gamma : \mathbb{R}^d \to \mathbb{R}$ defined by setting
\begin{equation}
\gamma(k) =  \sqrt{\om^2 + 4\sum_{j=1}^d \la_j\sin^2(k_j/2)} \, ,
\end{equation}
for some parameters $\omega >0$ and $\lambda_j >0$ for all $1 \leq j \leq d$.
Recall that 
\begin{equation}
\frac{ \partial \gamma}{ \partial k_j} (k) = \lambda_j \frac{ \sin(k_j)}{ \gamma(k)} \quad \mbox{for all } 1 \leq j \leq d \, ,
\end{equation}
it is clear that all points $k^*$ with components $k^*_j \in \{0, \pi \}$ for all $1 \leq j \leq d$ are critical points of $\gamma$.
Moreover, at each such point $k^*$, a short calculation shows that
 \begin{equation} \label{eq:gam2d}
\frac{\partial^2\gamma(k^*)}{\partial k_j\partial k_i} = \lambda_j \frac{ \cos(k^*_j)}{ \gamma(k^*)} \delta_0(i-j) \, ,
\end{equation}
and thus, each of these critical points is non-degenerate. In this case, the Morse Lemma
applies to $\gamma$ at each such critical point $k^*$. Let $U_{k^*}$ be the corresponding 
neighborhood and define $r(k^*)$ by taking the supremum over all $0<r \leq 1$ such that $B_r(k^*) \subset U_{k^*}$.
For the $2^d$ critical points mentioned above, take $2r = \min_{k^*}[r(k^*)]$.
This number $r>0$ will be used in our proof of Theorem~\ref{thm:xbd}.

%
%
%
\appendix
\section{Integral Inequalities}   \label{app:inineq}

In this section, we provide several estimates for
oscillatory integrals which are more or less well-known. 
We have included them to help demonstrate how our
results depend on certain parameters. For proofs, we 
follow closely the methods of \cite{Stein}.
The first few results, Lemma~\ref{l:outside_critical2} 
through Corollary~\ref{cor:kderbd}, apply to general phases 
for which lower bounds on derivatives are known.
Lemma~\ref{lem:prop5} generalizes Corollary~\ref{cor:kderbd}
to arbitrary multi-indices, but we present the result only in the
specific context needed in Section~\ref{sec:uniest}.    
The results contained in Lemma~\ref{lem:Qest} through Lemma~\ref{lem:gen}
concern the case that the phase is quadratic. They will be used
in Section~\ref{sec:main}.

We begin with some bounds applicable for general phases, i.e.
functions $\phi$ in the lemmas below. They apply when one has
lower bounds on derivatives of the phase. 
The first lemma is an analogue of Proposition 4 from Section 2 of Chapter VIII in \cite{Stein}.
In \cite{Stein}, the result is stated as an asymptotic. It is more useful to us as an inequality.

\begin{lem}                                                                     \label{l:outside_critical2}
Let $\phi : \mathbb{R}^d \to \mathbb{R}$ be smooth and take $\eta : \mathbb{R}^d \to \mathbb{C}$
smooth and of compact support. Suppose there exists a number $\alpha >0$ for which
\begin{equation} \label{eq:lbdphi1}
0< \alpha \leq \abs{ \frac{\partial}{ \partial y_j} \phi(y) } 
\end{equation}
for some $1 \leq j \leq d$ and all $y \in {\rm supp}( \eta)$.
Then, for any integer $N \geq 1$ and any real $t \neq 0$,
there exists a number $C_N(\phi, \eta)$ for which
\begin{equation} \label{eq:arbdec}
\abs{\int_{\R^d} e^{i t \phi(y)}
\eta(y) \,dy} \le \frac {C_{N}(\phi, \eta)} {|t|^N} \, .
\end{equation}
\end{lem}

As will be clear from the proof, the number $C_N(\phi, \eta)$ depends
on both $\alpha$ and $d$, but we do not include this in the notation. Moreover, the dependence of
$C_N( \phi, \eta)$ on $\phi$ and $\eta$ can be expressed in terms of quantities involving
only their partial derivatives with respect to $y_j$ up to order $N+1$.
  
\begin{proof}
Without loss of generality, we will assume that $j=1$. Fix $y \in {\rm supp}( \eta)$ and write
$y=(x, u)$ with $x \in \mathbb{R}$ and $u \in \mathbb{R}^{d-1}$. Regarding $u$
fixed, we will estimate the integral
\begin{equation}
I_1 (u) = \int_{\mathbb{R}} e^{it \phi(x, u)} \eta(x, u) \, dx \, 
\end{equation}
uniformly in $u$. With a slight abuse of notation, we will simply write the functions in the integrand above as $\phi(x)$ and $\eta(x)$.

Now, for any smooth function $f$ with $| \phi'(x)|>0$ for all $x \in {\rm supp}(f)$, we define
\begin{equation}
(D_{\phi} f)(x) = \frac{d}{dx} \left( \frac{f(x)}{\phi'(x)} \right) \, .
\end{equation} 
Using (\ref{eq:lbdphi1}), it is clear that $\eta$ is such a function.

Integration by parts yields
\begin{equation}
\abs{I_1} = \frac{1}{ |t|^N} \abs{\int_{\R} e^{i t \phi(x)}(D_{\phi}^N \eta)(x) \,dx} \, .
\end{equation}
An induction argument shows that for any $N \geq 1$,
\begin{equation}
(D_{\phi}^Nf)(k) = \sum_{k=0}^N \frac{f^{(k)}(x)}{( \phi'(x))^{2N-k}} P_{N,k}(x) \, ,
\end{equation}
where $P_{N,k}$ is a polynomial in the derivatives of $\phi$ up to order $N+1$. 
We have proven that
\begin{eqnarray}
\abs{\int_{\R^d} e^{i t \phi(y)} \eta(y) \,dy} & = & \abs{\int_{\R^{d-1}} I_1(u) \, d u} \nonumber \\
& \leq & \frac{1}{|t|^N} \sum_{k=0}^N \int_{\mathbb{R}^d} \alpha^{k-2N} \abs{ \left( \frac{ \partial}{ \partial y_1} \right)^k \eta(y)} \abs{P_{N,k}(y)} \, d \, y \nonumber \\ \label{eq:D4const}
& \leq &  \frac{C_N( \phi, \eta)}{| t|^N} \, ,
\end{eqnarray}
as claimed in (\ref{eq:arbdec}).
\end{proof}

The following one-dimensional result is a Corollary of Proposition 2 in Section 1 of Chapter VIII of \cite{Stein}.
We state it without proof and refer the interested reader to \cite{Stein}.

\begin{lem} \label{lem:1dest}
Let $\phi : \mathbb{R} \to \mathbb{R}$ be smooth and take $\eta : \mathbb{R} \to \mathbb{C}$ differentiable.
Suppose that for some integer $k \geq 2$ there is an $\alpha >0$ for which
\begin{equation}
0 < \alpha \leq \abs{ \phi^{(k)}(x)} \quad \mbox{for all } a \leq x \leq b \, . 
\end{equation}
Then,
\begin{equation}
\abs{ \int_a^b e^{i \phi(x)} \eta(x) dx} \leq  \frac{c_k}{ \sqrt[k]{\alpha}} \left( | \eta(b)| + \int_a^b \abs{ \eta'(x)} dx \right) \, .
\end{equation}
\end{lem} 
As is shown in \cite{Stein}, the number $c_k$ may be taken as $c_k = 5 \cdot 2^{k-1} -2$.
Moreover, in many applications of Lemma~\ref{lem:1dest}, ${\rm supp}(\eta) \subset (a,b)$ and so $\eta(b) =0$.
An immediate corollary follows.

\begin{cor} \label{cor:kderbd}
Let $\phi: \mathbb{R}^d \to \mathbb{R}$ have continuous partial derivatives up to order $k \geq 2$ and
let $\eta : \mathbb{R}^d \to \mathbb{C}$ have compact support contained in some open ball $B$.
If $\eta$ has continuous first order partial derivatives and there is some $1 \leq j \leq d$ and a number 
$\alpha>0$ for which   
\begin{equation}  \label{eq:d2phibd}
0 < \alpha \leq \abs{\frac {\partial^k \phi}{\partial y_j^k}(y)} \quad \mbox{for all} \quad y  \in B\, , 
\end{equation}
then 
\begin{equation} \label{eq:kdest}
\abs{\int_{\R^d} e^{i \phi(y)} \eta(y) \,dy} \le \frac{c_k}{ \sqrt[k]{ \alpha}} \int_{\mathbb{R}^d} \abs{ \frac{\partial \eta}{ \partial y_j} (y)} \, dy  ,
\end{equation}
with $c_k$ as in Lemma~\ref{lem:1dest}.
\end{cor}

\begin{proof}
Without loss of generality, we assume that $j=1$. Write each $y \in \mathbb{R}^d$ as $y=(x,u) \in \mathbb{R} \times \mathbb{R}^{d-1}$.
 Clearly,  
\begin{equation}
\int_{\R^d} e^{i \phi(y)} \eta(y) \,dy = \int_{\R^{d-1}} I_1(u) \,du
\end{equation}
where
\begin{equation}
I_1(u) = \int_{\mathbb{R}} e^{i \phi(x,u)} \eta(x, u) \, d x \, .
\end{equation}
Now, for any $u \in \mathbb{R}^{d-1}$ with $I_1(u) \neq 0$, there exists a finite interval
$(a(u), b(u))$ for which 
\begin{equation}
I_1(u) = \int_{a(u)}^{b(u)} e^{i \phi(x,u)} \eta(x, u) \, d x \, ,
\end{equation}
and ${\rm supp}(\eta( \cdot, u)) \subset (a(u), b(u))$. An application of Lemma~\ref{lem:1dest} shows that
\begin{equation}
\abs{I_1(u)} \leq \frac{c_k}{ \sqrt[k]{\alpha}} \int_{a(u)}^{b(u)} \abs{ \frac{\partial}{ \partial x} \eta(x,u)} \, dx \, ,
\end{equation}
and (\ref{eq:kdest}) readily follows.
\end{proof}

Given the above, we can now state and prove a modified version of Proposition 5 from Chapter VIII, Section 2 of \cite{Stein}
as needed in Section~\ref{sec:uniest}. Recall that we have introduced the function 
\begin{equation}
\phi_{t,x}(k) = k\cdot x  -2 t \gamma(k),
\end{equation} 
for each $t \in \mathbb{R}$ and $x \in \mathbb{Z}^d$ with
\begin{equation}
 \gamma(k) = \sqrt{\om^2 + 4\sum_{j=1}^d \la_j\sin^2(k_j/2)} \, 
\end{equation}
and parameters $\omega >0$ and $\lambda_j >0$ for all $1 \leq j \leq d$. 

\begin{lem} \label{lem:prop5}
Let $d \geq 2$, $m \geq 2$, and $\eta : \mathbb{R}^d \to \mathbb{C}$ have continuous first order partial derivatives
and compact support in some open ball $B$. Fix $t \in \mathbb{R} \setminus \{ 0 \}$ and $x \in \mathbb{Z}^d$. If there is a multi-index $\beta$ 
with $| \beta| = m$ and a number $\alpha >0$ for which 
\begin{equation} \label{eq:betaderbd}
0 < \alpha |t| \leq \abs{ \partial^{\beta} \phi_{t,x} (k) } \quad  \mbox{for all } k \in B \, ,
\end{equation}
then there exists a number $C$, independent of $t$ and $x$, for which
\begin{equation} \label{eq:prop5est}
\abs{\int_{\R^d} e^{i \phi_{t,x}(k)} \eta(k) \,dk} \le  \frac{C}{ \sqrt[m]{|t|}} \left( \| \nabla \eta \|_1 + \| \eta \|_1 \right) \, .
\end{equation}
\end{lem}

Here, the crucial observation is that the prefactor $C$ is independent of $x$. 
As can be seen from the proof, $C$ depends, in particular, on the diameter of $B$ and the distance from
the support of $\eta$ to the boundary of $B$. To be clear, we have denoted by
\begin{equation}                        \label{eq:11norm}
\| \nabla \eta \|_1 = \int_{\mathbb{R}^d} \sum_{j=1}^d \abs{ \frac{ \partial \eta}{ \partial k_j} (k)} \, dk \, .
\end{equation}

\begin{proof}
As is shown in Chapter VIII, Section 2.2 of \cite{Stein}, one can represent mixed partial 
derivatives as a linear combination of directional derivatives. In fact, with $\beta$ as above and
any smooth function $f$,
\begin{equation}
\partial^\beta f(k) = \sum_{ \ell =1}^D c_{\ell} (\xi_{\ell} \cdot \nabla)^m f (k)
\end{equation}
holds for some coefficients $c_{\ell}$, unit vectors $\xi_{\ell}$, and where $D = D(m)$ is the 
dimension of the space of all homogeneous polynomials of order $m$ in $\R^d$.
Now the bound \eqref{eq:betaderbd} implies that for each fixed $k_0 \in B$ there exists a unit vector $\xi$ such that
\begin{equation} \label{eq:lwbdkder}
|(\xi \cdot \nabla)^m \phi_{t,x} (k_0)| \ge \frac{\alpha |t| }{c D},
\end{equation}
where $c = \max_m |c_m|.$ In general, the unit vector $\xi$ will depend on $k_0$, $x$, and $t$.
We claim that, for each $k_0 \in {\rm supp}(\eta)$, there is a number $r>0$ for which
\begin{equation} \label{eq:lwbdkder2}
|(\xi \cdot \nabla)^m \phi_{t,x} (k)| \ge \frac{\alpha |t| }{2 c D} \quad \mbox{for all } k \in B_r(k_0) \, ,
\end{equation}
and moreover, we will show that the radius $r$ is independent of $k_0$, $x$, and $t$.

Let us denote by
\begin{equation}
F(k) = (\xi \cdot \nabla)^m \phi_{t,x}(k) \, .
\end{equation}
It is clear that
\begin{equation}
F(k) - F(k') = \sum_{j=1}^d \int_{k'_j}^{k_j} \partial_j F \left( (k_1', \cdots, k_{j-1}', s, k_{j+1}, \cdots, k_d) \right) \,  ds \, ,
\end{equation}
where we have denoted by $\partial_j$ the partial derivative with respect to $k_j$.
With this we find that for any $k \in B$,
\begin{equation}
\left| (\xi \cdot \nabla)^m \phi_{t,x}(k) - (\xi \cdot \nabla)^m \phi_{t,x}(k_0) \right| \leq 2 m d^2 \Gamma_{m+1} |t| |k-k_0| \, ,
\end{equation}
having set 
\begin{equation}
\Gamma_{m+1} = \max_{\beta': |\beta'| = m+1} \sup_{k \in \overline{B}} \abs{ \partial^{\beta'} \gamma(k)} < \infty \, .
\end{equation}
Taking 
\begin{equation}
r = \min \left[ {\rm dist}( {\rm supp}(\eta), B^c), \frac{\alpha}{4 cDmd^2 \Gamma_{m+1}} \right] \, ,
\end{equation}
we have proven (\ref{eq:lwbdkder2}).

With $r$ fixed as above, the compact support of $\eta$ can be covered with finitely many 
such open balls. We label these  $\{ B_{n} \}_{n=1}^{\mathcal{N}}$ and conclude $\mathcal{N}$ can be bounded
in terms of $r$, the diameter of $B$, and $d$. Choose a partition of unity $\{f_{n}\}$ subordinate to $\{ B_{n} \}$, i.e., 
choose $f_{ n} \in C^{\infty}(\mathbb{R}^d)$ with $0 \leq f_{ n}(y) \leq 1$, ${\rm supp}(f_{ n}) \subset B_{n}$, and 
\begin{equation}
\sum_{n=1}^N f_{ n}(y) = 1 \quad \mbox{for all} \quad y \in {\rm supp}( \eta) \, .
\end{equation}
It is clear then that
\begin{equation} 
\abs{\int_{\R^d} e^{i \phi_{t,x}(k)} \eta(k) \,dk} \le \sum_{n=1}^N \abs{\int_{\R^d} e^{i \phi_{t,x}(k)} \eta_{n}(k) \,dk} \, ,
\end{equation}
where we have set $\eta_{n} = f_{n} \eta$.

For each fixed $n$, we change variables in the corresponding integral above in such a way that
the first coordinate, say $y_1$, is aligned with the unit vector $\xi_n$ for which (\ref{eq:lwbdkder2}) holds on $B_n$. As a result of this rotation, the integral
\begin{equation}
\int_{\R^d} e^{i \phi_{t,x}(k)} \eta_n(k) \,dk = \int_{\R^d} e^{i \phi_{t,x}(k(y))} \eta_n(k(y)) \,dy
\end{equation}
satisfies the assumptions of Corollary~\ref{cor:kderbd} with $j=1$ and $\eta = f_n \eta$.  
We have proven that
\begin{equation}
\abs{\int_{\R^d} e^{i \phi_{t,x}(k)} \eta(k) \,dk} \le \frac{c_k \sqrt[k]{2cD}}{\sqrt[k]{\alpha |t|}} \left( N \left\| \frac{\partial}{ \partial y_1} \eta \right\|_1 + \sum_{n=1}^N \left\| \frac{\partial}{ \partial y_1} f_n \right\|_{\infty} \| \eta \|_1 \right) \, .
\end{equation}
The bound in (\ref{eq:prop5est}) follows by estimating the $y_1$-partial derivative with a gradient. 
\end{proof}

For the remainder of this appendix, we restrict our attention to the
special case that the phase is quadratic. More specifically, let
$Q: \mathbb{R}^d \to \mathbb{R}$ be given by
\begin{equation} \label{eq:quad}
Q(y) = \sum_{j=1}^d \epsilon_j y_j^2 \,
\end{equation}
where, for $1 \leq j \leq d$, the numbers $\epsilon_j \in \{ \pm 1 \}$. 
The following estimate results from a simple calculation. We state it for easy
reference.
\begin{lem} \label{lem:Qest}
Take $Q$ as in (\ref{eq:quad}) above. The bound
\begin{equation} \label{eq:Qest}
\abs{\int_{\R^d} e^{i t Q(y)} e^{-\abs{y}^2} dy}\le \left\{
\begin{array}{cc} (\pi)^{d/2} & \mbox{if } t= 0, \\ (\pi /
|t|)^{d/2} & \mbox{otherwise}, \end{array} \right.
\end{equation}
holds for each $t \in \mathbb{R}$.
\end{lem}

If the gaussian in (\ref{eq:Qest}) is multiplied by a smooth function that vanishes near the origin, then
the decay rate, in large $|t|$, improves quite a bit. The following notation will be convenient
for the remainder of this section. Let $f: \mathbb{R}^d \to \mathbb{C}$ be smooth. For any
$1 \leq j \leq d$ and any $N \geq 1$ set
\begin{equation}
\| f \|_{j, N} = \max_{0 \leq k \leq N} \| \partial_j^{(k)} f \|_{\infty} \, ,
\end{equation}
where we have denoted by $\partial_j^{(k)}$ the $k$-th partial derivative of $f$ 
with respect to the $j$-component.

\begin{lem} \label{lem:Qetaest}
Let $\eta \in C^{\infty}( \mathbb{R}^d)$ be a bounded function for
which, given any multi-index $\beta$, the function
$\partial^{\beta} \eta$ is bounded.  Suppose $\eta$ is supported away
from the origin, i.e.,
\begin{equation}						\label{eq:supp_s}
0<s = \inf_{y \in \supp (\eta)} |y| \, ,
\end{equation}
take $Q$ is as in (\ref{eq:quad}), and let $N \geq 1$ be an integer. There exists a number $C_N$, independent of $t$, such that
\begin{equation} \label{eq:Qetaest}
\abs{\int_{\R^d} e^{i t Q(y)} e^{-\abs{y}^2} \eta(y) dy}\le \frac
{C_N}{|t|^{N}} \sum_{j=1}^d \| \eta \|_{j,N} 
\end{equation}
holds for all $t \neq 0$.
\end{lem}
As will be clear in the proof below, the number $C_N$ depends on $d$ and $s$.
\begin{proof}
Following Chapter VIII, Section 2.3.1 of \cite{Stein}, we regard $\mathbb{R}^d$ as a union of cones.
For $1\leq j\leq d$, consider
\begin{equation}
\Gamma_j = \{y\in\R^d: 2 d y_j^2\ge \abs{y}^2\} \, .
\end{equation}
Clearly,  $\displaystyle \bigcup_{j=1}^d \Gamma_j = \R^d.$ Let  $\{\Omega_j\}_{j=1}^d$ be a collection of
functions,  homogeneous of degree 0 and smooth away from the origin, with
$\supp(\Omega_j) \subset \Gamma_j$ and satisfying
\begin{equation}
\sum_{j=1}^d \Omega_j(y) = 1, \quad \mbox{for all } y \ne 0.
\end{equation}
With respect to this partition of unity, we find that
\begin{equation}
\int_{\R^d} e^{i t Q(y)} e^{-\abs{y}^2} \eta(y) dy =
\sum_{j=1}^d \int_{\R^d} e^{i t Q(y)} e^{-\abs{y}^2} \eta_j(y) dy \, ,
\end{equation}
where we have written $\eta_j = \eta\Omega_j$.
We will estimate the integral corresponding to $j=1$. The others
are treated similarly.

For each $y \in {\rm supp}( \eta_1)$, we write $y =(x,u)$ with $x \in \mathbb{R}$ and $u \in \mathbb{R}^{d-1}$.
We will estimate the integral
\begin{equation}                                                                    \label{one_dim_int}
I_1(u) = \int_{\R} e^{ \pm i |t| x^2} e^{-x^2} \eta_1(x,u) dx \, ,
\end{equation}
uniformly with respect to $u \in \mathbb{R}^{d-1}$.

As in Lemma~\ref{l:outside_critical2}, the bound follows from integration by parts. 
With $u$ as above, it is clear that the function $\eta_1( \cdot, u)$ vanishes
near the origin. For any smooth function $f: \mathbb{R} \to \mathbb{C}$ supported
away from the origin, denote by $Df: \mathbb{R} \to \mathbb{C}$ the function given by
\begin{equation} \label{eq:defD}
(Df)(x) = \frac {d}{dx}\left(\frac {f(x)}{x}\right).
\end{equation}
Let $N \geq 1$ and integrate \eqref{one_dim_int} by parts $N$ times. The result is
\begin{align}                                                                                               \label{one_dim_int_byparts}
|I_1(u)| &= \frac {1}{(2 |t|)^N} \abs{\int_{\R} e^{\pm i |t| x^2}
D^N(e^{-x^2}\eta_1(x,u)) dx}.
\end{align}
By induction, one shows that
\begin{equation} \label{eq:DN}
D^N(e^{-x^2}\eta_1(x,u)) = e^{-x^2} \sum_{k=0}^N
\frac{\partial_1^{(N-k)}\eta_1(x,u)}{x^{N+k}} P_{k,N}(x) \, ,
\end{equation}
where $P_{k, N}$ is a polynomial of degree at most $2k$ with
coefficients that depend only on $k$ and $N$. 

Since $\sqrt{2d} \abs{x} \ge s$ on the support of $\eta_1$, we conclude two facts. First, for each $x \neq 0$ and $m \geq 1$,
\begin{eqnarray}
\abs{\partial_1^{(m)} \eta_1(x,u)} & \leq & \sum_{\ell =0}^m {m \choose \ell} \frac{\abs{ x^{\ell} \partial_1^{(\ell)} \Omega_1(x,u)}}{ |x|^{\ell}} \abs{ \partial_1^{(m- \ell)} \eta(x,u)} \nonumber \\
& \leq & C_m \| \eta \|_{1, m} \frac{(s + \sqrt{2d})^m}{s^m} \, .
\end{eqnarray}
Here we have used that $\Omega_1$ is smooth away from the origin, and homogeneous of degree zero, i.e. that $x^{\ell} \partial_1^{(\ell)} \Omega_1(x,u)$ is a bounded function.

Next, it is clear that
\begin{eqnarray}
\abs{I_1(u)} & \le & \frac{1}{(2|t|)^N} \sum_{k=0}^{N} \int_{\mathbb{R}} \frac{\abs{ \partial^{(N-k)} \eta_1(x,u)}}{|x|^{N+k}} e^{- x^2} \abs{P_{k,N}(x)} \, dx  \nonumber \\
& \leq & \frac{C_N}{|t|^N} \| \eta \|_{1,N} \, .
\end{eqnarray}
Since a similar estimate holds for each $1 \leq j \leq d$, we have
proven \eqref{eq:Qetaest}.
\end{proof}

If the gaussian in (\ref{eq:Qest}) is multiplied by a coordinate variable and a
smooth function with compact support containing the origin, then the decay
rate, in large $|t|$, can also be improved a bit.

\begin{lem} \label{lem:gen}
Let $\eta \in C_0^{\infty}( \mathbb{R}^d)$ with support containing a neighborhood of
zero. Take $Q$ as in (\ref{eq:quad}) and any $1 \le k \le d$. There exists a number $C$
such that 
\begin{equation} \label{eq:genlem}
\abs{\int_{\R^d} e^{i t Q(y)} e^{-\abs{y}^2} y_k \eta(y) dy} \le \frac{C}{|t|^{(d+1)/2}} \left( \| \eta \|_{\infty} + \sum_{j=1}^d \| \hat{\eta}_j \|_{j, d+2} \right)  
\end{equation}
holds for all $|t| \geq 1$. Here we have set $\hat{\eta}_j(y) = e^{- y_j^2} \eta(y)$.
\end{lem}
\begin{proof}
Again, we follow Chapter VIII, Section 2.3.1 of \cite{Stein}. Introduce $\Gamma_j$ and $\{ \Omega_j \}_{j=1}^d$ as in Lemma~\ref{lem:Qetaest} and write
\begin{equation}
\int_{\R^d} e^{i t Q(y)} e^{-\abs{y}^2} y_k \eta(y) dy =
\sum_{j=1}^d \int_{\R^d} e^{i t Q(y)} e^{-\abs{y}^2} y_k
\eta_j(y) dy,
\end{equation}
where again we have set $\eta_j = \eta \Omega_j$.

We bound the integral corresponding to $j=1$; the others
are treated similarly. To start, we insert a cut-off function 
$\alpha \in C_0^{\infty}(\mathbb{R}^d)$ with $\| \alpha\|_{\infty} \leq 1$ and
\begin{equation}
\alpha(y) = \left\{\begin{array}{cc} 1 & \mbox{if } |y| \leq 1, \\
0 & \mbox{if } |y| \geq 2 \, . \end{array} \right.
\end{equation}
For any $0 < \epsilon \leq 1$, set $\alpha_{\epsilon}(y) = \alpha(y/ \epsilon)$ and note that
\begin{eqnarray} \label{eq:Om1est}
\int_{\R^d} e^{i t Q(y)} e^{-\abs{y}^2} y_k \eta_1(y) dy &= & \int_{\R^d} e^{i t Q(y)} e^{-\abs{y}^2} y_k \eta_1(y) \alpha_{\epsilon}(y) dy \nonumber \\
& \mbox{ } & \quad + \int_{\R^d} e^{i t Q(y)} e^{-\abs{y}^2}
y_k \eta_1(y) (1 - \alpha_{\epsilon}(y))dy \, .
\end{eqnarray}
We will choose a value of $\epsilon$ later.

The first integral on the right hand side of (\ref{eq:Om1est}) 
satisfies
\begin{eqnarray} \label{eq:firstint}
\abs{\int_{\R^d} e^{i t Q(y)} e^{-\abs{y}^2} y_k \eta_1(y)\alpha_{\epsilon}(y) dy} & \le & \int_{\abs{y}\le 2\epsilon} \abs{y} \abs{\eta_1(y)}  \, dy \nonumber \\
& \leq & C \|  \eta \|_{\infty}
\epsilon^{d+1} \, ,
\end{eqnarray}
where $C$ depends on $d$ and $\| \Omega_1 \|_{\infty}$.

For the second integral on the right hand side of
(\ref{eq:Om1est}), we consider two cases. First, suppose that
$k \neq 1$. In this case, for each $y \in {\rm supp}(\eta_1)$ we 
write $y =(x, u)$ with $x \in \mathbb{R}$ and $u \in \mathbb{R}^{d-1}$.
The relevant one-dimensional integral is 
\begin{equation} \label{eq:defj1}
I_1(u) = \int_{\R} e^{\pm i |t| x^2} e^{- x^2}  \eta_1(x,u) \left( 1- \alpha_{\epsilon}(x, u) \right) \, d x \, .
\end{equation}
Again, we will integrate by parts. Let $D$ be as defined in \eqref{eq:defD}. 
We need an analogue of (\ref{eq:DN}). Observe that if
$f,g: \mathbb{R} \to \mathbb{C}$ are smooth functions, whose
product vanishes in a neighborhood of the origin, then for every integer $N \geq 1$,
\begin{equation} \label{eq:DNfg}
[D^N(fg)](x) = \sum_{n=0}^N F_N(x;n) \frac{g_u^{(N-n)}(x)}{x^{N+n}}
\, ,
\end{equation}
where
\begin{equation} \label{eq:Fcoef}
F_N(x;n) = \sum_{j=0}^n F_{N,n}(j) x^j f^{(j)}(x) \, .
\end{equation}
Here $F_{N,n}(j)$ are numerical coefficients that are independent of the functions $f$ and $g$.
We apply this formula with $f_u(x) = \Omega_1(x,u)$ and $g_u(x) = e^{-x^2} \eta(x,u) (1- \alpha_{\epsilon}(x,u))$.
Note that the product of $f_u$ and $g_u$ is supported on those $x$ for which $2dx^2 \geq \epsilon^2$ uniformly
with respect to $u$. In what follows, we will suppress the dependence of $f_u$ and $g_u$ on $u$.

Integration by parts yields
\begin{eqnarray} \label{eq:j1ibp}
\abs{I_1(u)} & = & \frac{1}{(2|t|)^N} \abs{ \int_{\R} e^{\pm i|t|x^2} \left[ D^N \left(f g \right) \right](x) \, d x} \nonumber \\
& \leq & \frac{1}{(2|t|)^N} \sum_{n=0}^N \int_{\R} \abs{F_N(x;n)} \, \frac{ \abs{g^{(N-n)}(x)}}{|x|^{N+n}} \, d x \, ,
\end{eqnarray}
and therefore
\begin{equation*}
\abs{ \int_{\R^d} e^{i t Q(y)} e^{-\abs{y}^2} y_k \eta_1(y) (1 - \alpha_{\epsilon}(y))dy} \, \leq  \frac{1}{(2|t|)^N} \sum_{n=0}^N \int_{\R^{d-1}} |u_k|  \, \int_{\R} \abs{F_N(x;n)} \, \frac{ \abs{g^{(N-n)}(x)}}{|x|^{N+n}} \, d x \, d u \, , \end{equation*}
where we have set $u_k = y_k$. 

Now, since $\mbox{supp}( \Omega_1) \subset \Gamma_1$, it is clear that
\begin{equation}
|x|^2 \geq \frac{|y|^2}{2d} \quad \Rightarrow \quad \frac{1}{|x|^{N+n}} \leq \frac{(\sqrt{2d})^{N+n}}{|y|^{N+n}} \, .
\end{equation}
Moreover, as $\Omega_1$ is smooth (away from 0) and homogeneous of degree 0, the function $F_N$ is bounded.

{F}rom these observations, we see that for any $0 \leq n \leq N$,
\begin{equation*}
\int_{\R^{d-1}} |u_k|  \, \int_{\R} \abs{F_N(x;n)} \, \frac{ \abs{g^{(N-n)}(x)}}{|x|^{N+n}} \, d x \, d u \leq 
C_{N} \int_{\R^{d-1}} |u_k|  \, \int_{\R} \frac{ \abs{g^{(N-n)}(x)}}{|(x,u)|^{N+n}} \, d x \, d u \, ,
\end{equation*}
where the final integral is well defined since $g$ is compactly supported in $B_{\epsilon}(0)^c$.

We bound the derivatives of $g$ using the product rule. Clearly, for any
integer $k \geq 1$,
\begin{eqnarray}
\abs{g^{(k)}(x) } & \leq & \sum_{j=0}^k  {k \choose j} \abs{  \partial_1^{(j)}\left(1- \alpha_{\epsilon}(x,u) \right)}  \abs{ \partial_1^{(k-j)} \hat{\eta}(x,u)} \nonumber \\
& \leq & 2 \| \partial_1^{(k)} \hat{\eta} \|_{\infty} + \sum_{j=1}^k {k \choose j} \epsilon^{-j} \| \partial_1^{(j)} \alpha \|_{\infty} \| \partial_1^{(k-j)} \hat{\eta} \|_{\infty} \, ,
\end{eqnarray} 
where we have set $\hat{\eta}(y) = e^{-y_1^2} \eta(y)$. In this case, with $0 \leq k \leq N$ the estimate
\begin{equation}
\| g^{(k)} \|_{\infty} \leq 2^{k+1} \| \alpha \|_{1,N} \| \hat{\eta} \|_{1,N} \epsilon^{-k} \, ,
\end{equation}
is valid for any $\epsilon \leq 1$.
As a result, we have that
\begin{eqnarray} \label{eq:secint}
\abs{ \int_{\R^d} e^{i t Q(y)} e^{-\abs{y}^2} y_k \eta_1(y) (1 - \alpha_{\epsilon}(y))dy} & \leq  &
\frac{1}{(2|t|)^N} \sum_{n=0}^N C_{N, n} \| g^{(N-n)} \|_{\infty} \int_{|y| \geq \epsilon} |y|^{1-N-n} \, d y \nonumber \\
& \leq & \frac{C_N}{|t|^N} \| \hat{\eta} \|_{1,N}  \epsilon^{d+1-2N}  \, ,
\end{eqnarray}
for any $N > d+1$ and $\epsilon \leq 1$.

Combining (\ref{eq:firstint}) and (\ref{eq:secint}), we have shown that
\begin{equation} \label{eq:K1K2}
\abs{\int_{\R^d} e^{i t Q(y)} e^{-\abs{y}^2} y_k \eta_1(y) dy} \leq \epsilon^{d+1} \left( C \| \eta \|_{\infty} +  \frac{C_N}{( \epsilon^2 |t| )^N} \| \hat{\eta} \|_{1,N} \right) \, , 
\end{equation}
for any $N > d+1$ and $\epsilon \leq 1$. Taking $|t| \geq 1$, $\epsilon = |t|^{-1/2}$, and $N = d+2$, we find an estimate of the form (\ref{eq:genlem}).
This completes the first case.

If $k =1$, we proceed as above using that
\begin{equation} \label{eq:DNfxg}
[D^N(f(x) x g(x))](x) = \sum_{k=0}^N F_N(x;k)
\frac{g^{(N-k)}(x)}{x^{N-1+k}} \, ,
\end{equation}
where
\begin{equation} \label{eq:Fcoef2}
F_N(x;k) = \sum_{j=0}^k F_{N,k}(j) x^j f^{(j)}(x) \, .
\end{equation}
(\ref{eq:genlem}) follows using the methods above.
\end{proof}

\end{document}